\documentclass[a4paper,12pt,oneside,final]{article}
\usepackage[utf8]{inputenc}  
\usepackage[english]{babel}
\usepackage{setspace}
\usepackage[natbibapa]{apacite}   
\usepackage{natbib}

\usepackage{amssymb,amsmath,mathrsfs,amsthm}
\usepackage{amsfonts}
\usepackage{mathtools}
\usepackage{euscript} 
\usepackage{verbatim}

\usepackage[font=small]{caption}
\usepackage{subcaption}
\usepackage{sectsty}	

\usepackage{eurosym}
\usepackage{booktabs}
\usepackage{array}
\usepackage{makecell}
\usepackage{tabularx}
\usepackage{enumitem}
\usepackage{microtype}
\usepackage{dsfont}
\usepackage{csquotes}
\usepackage[bottom]{footmisc}

\usepackage{geometry}
\usepackage[shadow,colorinlistoftodos]{todonotes}

\usepackage{thmtools,thm-restate} 

\newtheorem{proposition}{Proposition}
\newtheorem{corollary}{Corollary}
\newtheorem{lemma}{Lemma}
\newtheorem{definition}{Definition}

\newtheorem{step}{Step}

\newtheorem{Alemma}{Lemma}[section] 
\newtheorem{Aproposition}{Proposition}[section] 
\newtheorem{Acorollary}{Corollary}[section] 

\DeclareMathOperator*{\argmax}{arg\,max}
\DeclareMathOperator*{\argmin}{arg\,min}
\definecolor{coolblack}{rgb}{0.0, 0.18, 0.39}

\makeatletter
\def\@fnsymbol#1{\ensuremath{\ifcase#1\or \dagger\or \ddagger\or
		\mathsection\or \mathparagraph\or \|\or **\or \dagger\dagger
		\or \ddagger\ddagger \else\@ctrerr\fi}}
\makeatother

\usepackage[colorlinks=true,linkcolor=coolblack,
urlcolor=coolblack,
citecolor=coolblack,pdftex]{hyperref}
\usepackage[all]{hypcap}

\begin{document}
	
	\pagestyle{plain}
	
	\author{Federico Vaccari\thanks{Department of Economics and Management, University of Trento, I-38122 Trento, Italy. E-mail: \href{mailto:vaccari.econ@gmail.com}{\sf vaccari.econ@gmail.com}. Santiago Oliveros provided invaluable guidance, support, and help. For helpful comments and suggestions, I thank Yair Antler, Leonardo Boncinelli, Luca Ferrari, Christian Ghiglino, Johannes Hörner,  Aniol Llorente-Saguer, Simon Lodato, Marina G. Petrova and participants at the Conference on Mechanism and Institution Design 2020, Formal Theory Virtual Workshop, Royal Economic Society Annual Conference 2021, GRASS XIV Workshop, 61st Annual Conference of the Italian Economic Association, and the 12th World Congress of the Econometric Society. All errors are mine. This project has received funding from the European Union's Horizon 2020 Research and Innovation Programme (Marie Sk\l odowska-Curie grant no.~843315-PEMB).}}
	
	\title{Influential News and Policy-making} 
	
	\date{} 
	\maketitle

	\begin{abstract}
		\noindent It is believed that interventions that change the media’s costs of misreporting can increase the information provided by media outlets. This paper analyzes the validity of this claim and the welfare implications of those types of interventions that affect misreporting costs. I study a model of communication between an uninformed voter and a media outlet that knows the quality of two competing candidates. The alternatives available to the voter are endogenously championed by the two candidates. I show that higher costs may lead to more misreporting and persuasion, whereas low costs result in full revelation; interventions that increase misreporting costs never harm the voter, but those that do so slightly may be wasteful of public resources. I conclude that intuitions derived from the interaction between the media and voters, without incorporating the candidates’ strategic responses to the media environment, do not capture properly the effects of these types of interventions.
	\end{abstract}
	
	
	
	\thispagestyle{empty}
	
	\clearpage
	
	\tableofcontents
	
	\pagebreak

\section{Introduction}

One of the most common criticisms leveled against the media is that they strategically distort news to pursue their private interests and affect political outcomes.\footnote{This concern is substantiated by empirical evidence that media bias has an impact on voting behavior (see, e.g., \citealp{dellavigna2007}) and by the observation that mass media are voters' primary source of policy-relevant information (see \citealp{pew2016}).} To counter the threat posed by the spread of misinformation, most countries enforce laws that punish the practice of misreporting information. Consider for example the United Kingdom's Representation of the People Act 1983 (Chapter 2, Part II, Section 106):
\begin{displayquote}
	\small
	A person who, or any director of any body [...] which --
	(a) before or during an election,
	(b) for the purpose of affecting the return of any candidate at the election,
	makes or publishes any false statement of fact in relation to the candidate's personal character or conduct shall be guilty of an illegal practice.
\end{displayquote}
More recently, several governments have passed ``fake news laws'' to address the growing concern about distortions of the political process caused by misinformation. Most of these efforts revolve around the idea of affecting media outlets' costs of misreporting information through, e.g., fines, jail terms, and awareness campaigns \citep{funke2018guide}.\footnote{Misreporting costs can be direct, such as the time and money required to misrepresent information, or indirect and probabilistic, such as the loss of reputation and profits incurred by a media outlet if caught in a lie. For more examples about ``fake news laws,'' see \cite{funke2018guide}.}

This class of interventions is relevant not only because of its recent popularity, but also because it seeks to steer the conduct of media outlets without interfering with the markets' concentration levels. In ``news markets'' a single outlet with private possession of some information is in fact a monopolist over that particular piece of news. This is often the case with scoops, scandals, and ``October surprises.'' Since breaking news spreads fast, even small outlets can reach a large audience when endowed with a scoop that can swing the outcome of an election.\footnote{In the ``Killian document controversy,'' online blogs' revelation that CBS aired unauthenticated and forged documents was quickly rebroadcast by a wide spectrum of media. See \cite{gentzkow2008} for this and other examples.}  In these circumstances, interventions that affect the costs of misreporting information might still discipline the behavior of those media outlets with exclusive possession of political news. Despite its relevance, the regulation of misreporting costs is currently highly understudied, and to date there is no formal model exploring its consequences.\footnote{This is because most related work assumes that misreporting is either costless (e.g., in cheap talk models) or not possible (e.g., in disclosure models). See Section~\ref{sec:lit} for a review of the relevant literature.}

In this paper, I study the welfare effects of regulatory interventions that impose costs on media outlets for misreporting information. The key idea is that the implications of media bias are not confined to distortions of voters' choice at the ballot box, but spread and propagate back to the process of policy-making. Ahead of elections, competing candidates face a choice between gathering popular consensus with ``populist'' policies that benefit voters and seeking the endorsement of an influential media with ``biased'' policies that please media outlets. Since media bias skews electoral competition and produces distortions in policy outcomes, the informational and political effects of regulation need to be jointly determined.

I consider a model of strategic communication between a media outlet and a representative voter, where the policy alternatives available to the voter are endogenously championed by two competing candidates running for office. In the policy-making stage, the two candidates -- an incumbent and a challenger -- sequentially and publicly make a binding commitment to policy proposals. Afterwards, in the communication subgame, the media outlet delivers a public news report about the candidates' relative quality, where ``quality" is defined as a candidate's fit with the state of the world, competence, record, etc. Given the proposals and the outlet's report, the voter casts a ballot for one of the two candidates. At the end, the policy proposed by the elected candidate gets implemented. 

In contrast to canonical models of strategic communication, the media outlet bears a cost of misreporting its private information about candidates' quality  that is increasing in the magnitude of misrepresentation. The voter and the outlet have aligned preferences over the relative quality of candidates (hereafter just ``quality''), but disagree on which policy is the best. Therefore, when candidates advance different proposals, there are contingencies in which there is a conflict of interest between the outlet and the voter. An agency problem emerges, as the outlet can strategically misreport information to induce the election of its favorite candidate and seize political gains at the expense of the voter.

The main results provide a number of policy implications by showing how the regulation of misreporting costs affects the voter's welfare. I find that an increase in the costs of misreporting information never harms the voter, but that a small incremental increase might have no effect at all on the voter's welfare. This result implies that, when carrying out interventions is costly, lenient regulatory efforts can be wasteful of public resources. I obtain conditions under which the voter is better off without a media outlet or alternatively with an ``electoral silence'' period that forbids the delivery of policy-relevant news ahead of the election.\footnote{Some countries operate a pre-election silence period where even polling and campaigning are not allowed in the run-up to an election, while in other countries such bans are unconstitutional.} I also show that there is no monotonic relationship between the probability that persuasion takes place and the voter's welfare, and between the probability that persuasion takes place and the costs of misreporting information. Interventions that increase such costs might induce more misreporting and more persuasion, and yet improve the voter's welfare because of the availability of better policies. Therefore, the growing concern that ``proposed anti-fake news laws [...] aggravate the root causes fuelling the fake news phenomenon" \citep{alemanno_2018} is perhaps exaggerated. This also implies that the empirical task of inferring the efficiency of such interventions from the media's reporting behavior is challenging, if not impossible.

A natural question is whether politicians have the right incentives to propose interventions that benefit the voter.\footnote{Most fake news laws are introduced by members of incumbent governments, ministers, or government factions \citep{lawcongress2019,funke2018guide}.} To answer this question, I extend the main model by endogenizing the process of regulating misreporting costs, which takes place ahead of the policy-making stage. I show that the electoral incentives of politicians together with the sequential nature of policy proposals generate a friction in the regulatory process that results in the selection of interventions that depress the voter's welfare. The worst-case scenario is obtained when the incumbent government is in charge of regulation: in this case, the incumbent sets relatively low misreporting costs that trigger the convergence of proposals to the media outlet's favorite policy. Even though misreporting behavior is fully eradicated, the voter's welfare is at its minimum because of the induced policy distortion. From the voter's perspective, the resulting political outcome is abysmal, and equivalent to that of a dystopic scenario where the media outlet has the voting rights to directly decide upon which policy to implement and which candidate to elect. The situation is better, but still far from ideal, when the challenger is in charge of regulation.

The intuition behind the above results is as follows. As the costs of misreporting information decrease, both candidates offer more ``biased'' policies in the attempt to obtain the endorsement of an increasingly persuasive media outlet. The candidates' proposals become progressively closer to each other until, for sufficiently low misreporting costs, they fully converge on the outlet's preferred policy. More similar policies imply a smaller conflict of interest between the voter and the media outlet, and thus persuasion can occur in a smaller number of contingencies as costs decrease. Eventually, the convergence of proposals eradicates any conflict of interest as in these cases the only element that can differentiate candidates is their relative quality, over which preferences are aligned. Somewhat paradoxically, under low misreporting costs the media outlet has high persuasive potential and yet it fully reveals its private information about quality. However, the voter's welfare is at its minimum because perfect knowledge about quality---and thus perfect selection of candidates---comes at the cost of inducing a large distortion in terms of policies, which are furthest from the voter's ideal bliss policy. If candidates' quality is sufficiently less important than their policies, then the voter might be better off without a media outlet at all.

Since policy convergence occurs for a set of sufficiently low but positive misreporting costs, lenient interventions might be ineffective. On the other hand, a substantial increase in the misreporting costs might trigger policy divergence and thus increase the contingencies in which there is a conflict of interest, making room for more misreporting and persuasion. In these cases, the voter's welfare increases because the loss of information about quality and the increased electoral mistakes are more than compensated by the availability of better policies. When misreporting costs are sufficiently high, both candidates offer more ``populist'' policies to please the voter rather than the weakened media outlet. As costs increase, the candidates' proposals tend to converge back toward the voter's preferred policy, mitigating the conflict of interest and the occurrence of misreporting and persuasion. The voter's welfare is thus maximized by arbitrarily high misreporting costs.

To see how electoral incentives skew the process of regulation, recall that policies are proposed sequentially. The presence of an influential media outlet transforms the policy-making stage in a sort of sequential rock-paper-scissors game where a moderate policy beats a populist one, a biased policy beats a moderate one, and a populist policy beats a biased one. Given the incumbent's proposal, the challenger has the second-mover advantage of choosing the most profitable strategy between seeking the voter's approval or the media outlet's support. When in charge of regulation, the incumbent can nullify this ``incumbency disadvantage effect'' by setting low misreporting costs to force policy convergence.\footnote{There is empirical evidence of both media's anti-incumbent behavior \citep{puglisi2011being} and of an ``incumbency disadvantage'' effect due to media coverage \citep{green2017}. However, evidence is mixed as other work finds that media has either no clear effect \citep{gentzkow2011effect} or a positive effect on the reelection chances of incumbent politicians \citep{drago2014meet}.} Therefore, electoral incentives push politicians to use regulation for purely instrumental reasons, decreasing the voter's welfare as a result.

The remainder of this article is organized as follows. In Section~\ref{sec:lit}, I discuss the related literature. Section~\ref{sec:model} introduces the model, which I solve in Section~\ref{sec:eqm}. In Section~\ref{sec:welfare}, I analyze the voter's welfare and the process of regulation. In Section~\ref{sec:extensions}, I discuss the model and its extensions, and Section~\ref{sec:conclusion} concludes. Formal proofs are relegated to the Appendix.

\section{Related Literature}\label{sec:lit}
This paper is related to the literature studying the political economy of media bias.\footnote{For comprehensive surveys on the topic, see \cite{prat2013} and \cite{gentzkow2015media}.} Papers belonging to this literature can be broadly split into two strands: models of demand-side and models of supply-side media bias. The first strand focuses on the case where news organizations are profit-maximizing and/or their preferences over political outcomes are second-order. Bias can emerge, for example, when media firms and journalists want to develop a reputation for accurate reporting \citep{gentzkow2006,shapiro2016}, consumers favor confirmatory news \citep{mullainathan2005,bernhardt2008}, or voters demand biased information \citep{oliveros2015demand,calvert1985,suen2004}. In the present paper I take a supply-side approach by considering a media outlet that has preferences over political outcomes. In this second strand, bias originates from the intrinsic preferences and motivations of agents who work for news organizations, like editors and owners. For example, media bias occurs when journalists have ideological leanings \citep{baron2006}, media firms suppress unwelcome news \citep{besley2006,anderson2012}, or politicians design public signals \citep{alonso2016}.   

The above-mentioned papers abstract from the process of policy-making and political competition. By contrast, I explicitly incorporate an electoral stage where candidates compete via binding commitments to policy proposals. For this reason, the present paper is more closely related to the stream of work studying the effects of political endorsements on policy outcomes. Within this part of the literature but in contrast to the present paper, \cite{grossman1999competing}, \cite{gul2012}, and \cite{chakraborty2019expert} consider voters that are uncertain about their own preferences; \cite{carrillo2008information} and \cite{boleslavsky2015information} model the source of information about candidates as exogenous; \cite{andina2006political} models voting behavior as exogenous; \cite{miura2019manipulated} considers a media outlet that delivers fully certifiable information about candidates' policies; \cite{chan2008} and \cite{stromberg2004} study a demand-side framework; \cite{ashworth2010} and \cite{warren2012} use a political agency framework to study how a media outlet affects the incumbent's incentives to pander.

The most closely related paper is \cite{chakraborty2016}. They use a Downsian framework to study the welfare effects of a policy-motivated media outlet that can influence voting behavior via cheap talk endorsements. The present paper is different in three important aspects: first, I incorporate costs for misreporting information that are proportional to the magnitude of misrepresentation. Under this approach, a news report is more than just an endorsement as it constitutes a costly signal of the state (on this point, see also the next paragraph). Second, I study a sequential rather than a simultaneous model of electoral competition. As a result, I obtain that the policy of the incumbent is subject to a different distortion than that of the challenger. I show that this difference plays an important role when endogenizing the process of regulation. Finally, the welfare analysis in \cite{chakraborty2016} focuses on the ideological conflict between the media outlet and the voter, while I focus on the intensity of misreporting costs and its regulation.

The key feature of the present paper is how communication is modeled. Papers in the previously mentioned literature consider media outlets that either can report anything without bearing any direct consequence on their payoffs (e.g., \citealp{gul2012,chakraborty2016}) or cannot misreport information at all (e.g., \citealp{duggan2011,besley2006}). By contrast, I consider a media outlet that can misreport information but at a cost. In addition to incorporating a realistic feature, this modeling strategy allows me to perform comparative statics on misreporting costs that are currently unexplored, yet crucial for understanding the regulation of news markets.

Therefore, the present paper also touches upon the literature on strategic communication with lying costs \citep{kartik2007,kartik2009,ottaviani2006,chen2011}. With respect to this line of work, I consider a setting where the voter (i.e., the receiver) has a binary action space and the outlet (i.e., the sender) has a continuous message space. Moreover, the alternatives available to the voter are endogenously selected through a process of electoral competition, and not exogenously given. This framework gives rise to a number of important qualitative differences in the amount of information transmitted and the language used in equilibrium: I obtain equilibria where persuasion naturally occurs even within a large state space; the sender might invest costly resources to misreport even in the (interim) absence of a conflict of interest with the receiver; full information revelation occurs with relatively low misreporting costs.\footnote{With a coarse action space, the outlet can achieve persuasion by pooling information to make the voter indifferent between two actions. Similarly, \cite{chen2011} obtains ``message clustering'' in a setting with a continuous action space and a coarse message space. \cite{kartik2009} finds partial separation in a bounded type space setting. \cite{kartik2007} and \cite{ottaviani2006} show that full separation is achieved when such a bound is arbitrarily large.} These features are key and instrumental for the main results of the present paper.

\section{The Model}\label{sec:model}

There are four players: a representative voter $v$, a media outlet $m$, and two candidates: an incumbent $i$, and a challenger $c$. The voter has to cast a ballot $b\in\{i,c\}$ for one of the two candidates. At the outset, in the ``policy-making stage,'' each candidate makes a binding and public commitment to a policy proposal. I assume that proposals are sequential: the incumbent first commits to a policy $q_i\in\mathbb{R}$; after observing $q_i$, the challenger commits\footnote{The assumption of sequentiality in the policy-making process reflects that candidates announce their policies at distinct points in time, and that the incumbent's policy is typically formed or known before the challenger's. See, e.g., \cite{wiseman2006}.} to a policy $q_c\in\mathbb{R}$. Policy proposals $q=(q_i,q_c)$ are then publicly observed by all players. If the voter casts a ballot for candidate $j\in\{i,c\}$, then policy $q_j$ is eventually implemented. 

The ``communication subgame'' takes place after the candidates' commitments but before the election: the media outlet privately observes the realization of a state $\theta\in\Theta$ and then delivers a news report $r\in\mathbb{R}$. Reports are literal statements about the state. Before casting a ballot, the voter observes the report $r$ but not the state $\theta$. Figure~\ref{fig:timeline} illustrates the timing structure of the model.

\begin{figure}
	\centering
	\includegraphics[width=\linewidth]{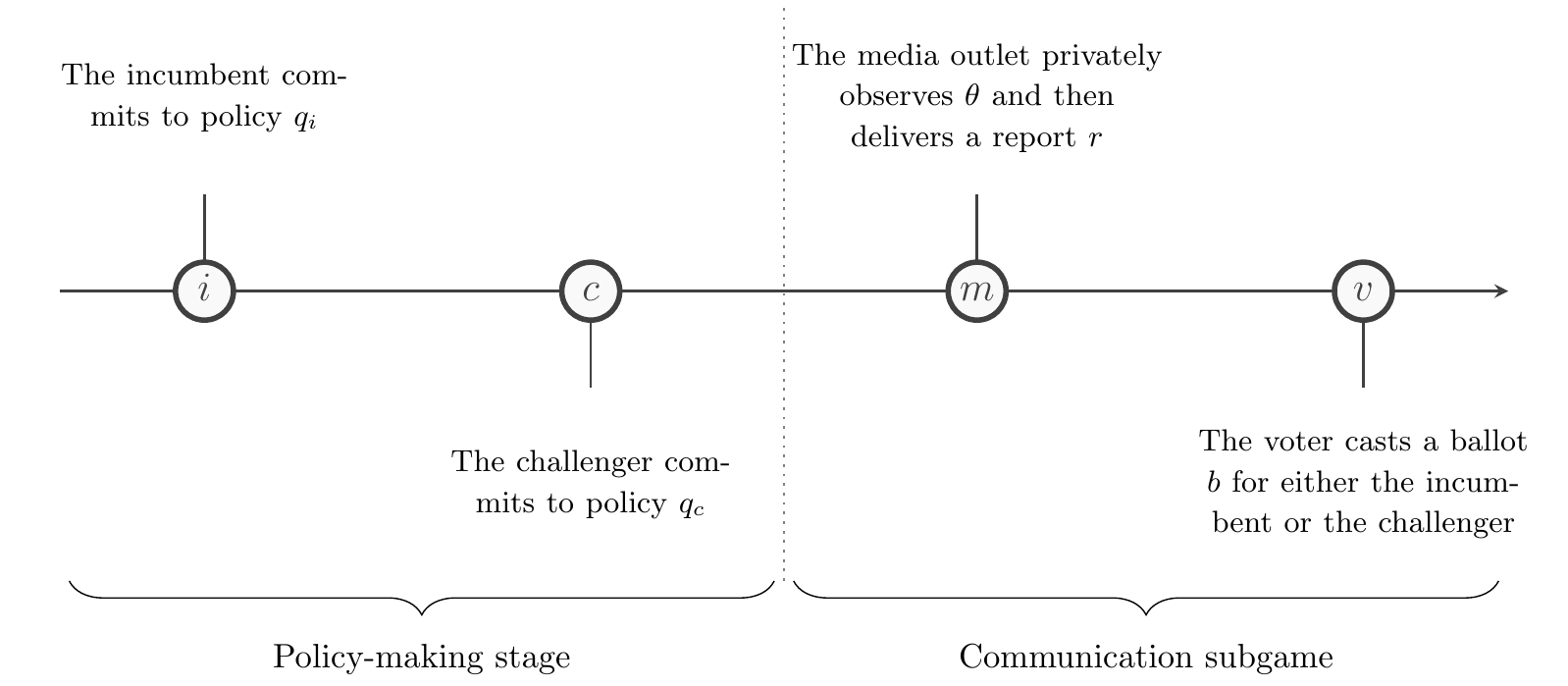}
	\caption{Timeline of the model.}
	\label{fig:timeline}
\end{figure}

{\bf The State.} The state $\theta$ represents the relative quality of the incumbent with respect to the challenger, and I shall hereafter refer to $\theta$ simply as ``quality.'' I assume that $\theta$ is randomly drawn from a uniform density function $f$ over $\Theta=[-\phi,\phi]$, where $f$ is common knowledge to all the players. Only the media outlet privately observes the realized $\theta$. The voter and the media outlet have identical preferences over quality: given any proposals $q=(q_i,q_c)$, the higher the quality is, the higher is the gain from the incumbent winning the election rather than the challenger. Thus, quality is an element of vertical differentiation similar in kind to what is known in political theory as ``valence.''\footnote{Quality can capture traits like a candidate's fit with the state of the world, competence, reputation, evidence of good and bad conduct, etc. For a discussion of the closely related notion of valence or character, see, e.g., \cite{stokes1963}, \cite{kartik2007signaling}, and \cite{chakraborty2016}.}

{\bf Payoffs.} Candidates are purely office-seeking, and care only about their own electoral victory. I assume that candidates obtain a utility of 1 if they win and 0 otherwise. The utility of candidate $j\in\{i,c\}$ is thus\footnote{$\mathds{1}\{\cdot\}$ is the indicator function, where $\mathds{1}\{A\}=1$ if $A$ is true, and $0$ otherwise.} $u_j(b)=\mathds{1}{\{b=j\}}$. 

The voter and the media outlet have a preferred ``bliss'' policy of $\varphi_v\in\mathbb{R}$ and $\varphi_m\in\mathbb{R}$ respectively.\footnote{The model can allow for the presence of a finite committee or a continuum of voters where $v$ is the median voter with bliss policy $\varphi_v$. Under a majority voting rule with two alternatives, the assumption of sincere voting is without loss of generality, as in such cases truth-telling is a dominant strategy.} I assume without loss of generality that $\varphi_m<\varphi_v$, and denote with $\gamma>0$ a parameter weighting the relative importance of policies to quality. The voter's utility $u_v(b,\theta,q)$ from selecting candidate $b\in\{i,c\}$ when quality is $\theta$ and proposals are $q=(q_i,q_c)$ is an additively separable combination of standard single-peaked policy preferences and quality, i.e.,
\[
u_v(b,\theta,q)=-\gamma (\varphi_v-q_b)^2 + \mathds{1}{\{b=i\}} \theta.
\]
Therefore, given proposals $q$, the voter prefers to vote for the incumbent only if quality is high enough, i.e., $\theta>\tau_v(q)=\gamma(2\varphi_v -q_c-q_i)(q_c-q_i)$. The threshold $\tau_v(q)$ is obtained from solving $u_v(i,\theta,q)=u_v(c,\theta,q)$ for $\theta$. 

I similarly define $\tau_m(q)=\gamma(2\varphi_m -q_c-q_i)(q_c-q_i)$ and refer to $\tau_j(q)$ as player $j$'s threshold, for $j\in\{v,m\}$. The media outlet's endorsed candidate is\footnote{More precisely, $\hat m(\theta,q)$ is the outlet's preferred candidate given policies $q$ and state $\theta$. Since in equilibrium the outlet will endorse her preferred candidate, I refer to $\hat m(\cdot)$ as the ``endorsed candidate.''}
\[
	\hat m(\theta,q)=
		\begin{cases}
			i  & \quad \text{if } \theta> \tau_m(q) \\
			c  & \quad  \text{ otherwise.}
		\end{cases}
\] 
I denote by $k>0$ a scalar parameter measuring the intensity of misreporting costs, and by $\xi>0$ the outlet's gains from endorsing the candidate that ends up winning the election.\footnote{Alternatively, $\xi$ may indicate the avoided losses that the outlet would expect to incur when opposing the winning candidate. For example, after the publication of the first stories about the Watergate scandal, President Nixon allegedly said ``The \emph{Post} is	going to have damnable, damnable problems out of this one. They have a television station. \ldots and they're going to have to get it renewed.'' See references in \citet[p.~136]{gentzkow2008}.} The media outlet gets a utility of $u_m(r,b,\theta,q)$ when delivering report $r$ in state $\theta$ with proposals $q$ and winning candidate $b$, where
\[
u_m(r,b,\theta,q)=\mathds{1}{\left\{b=\hat m(\theta,q)\right\}} \xi -k(r-\theta)^2.
\]
Unlike the voter, the outlet's utility depends on whether the endorsed candidate $\hat m(\theta,q)$ is elected, but not on the implemented policy $q_b$. This assumption allows me to model a media outlet whose endorsements depend on the candidates' policies even when the policies do not directly affect the outlet's payoff. This is often the case, for example, when editors and journalists have political leanings on issues such as abortion or gay marriage that have no direct impact on the media organization itself. The score $\xi$ represents the outlet's benefits from endorsing the victorious rather than the defeated candidate. In Section~\ref{sec:extensions}, I show that most qualitative results hold when considering political gains that depend on policies $q$ and bliss $\varphi_m$ in a similar way as the voter's.

In addition, the media outlet incurs a cost of $k(r-\theta)^2$ for delivering a news report $r$ when the state is $\theta$. Any report $r\in\mathbb{R}$ has the literal or exogenous meaning ``quality is equal to $r$.'' Truthful reporting occurs when $r=\theta$, and it is assumed to be costless. By contrast, misreporting information is costly, and the associated costs are increasing with the difference between the stated and the true realization of quality. The score $k$ encapsulates all those elements determining the magnitude of misreporting costs, such as reputation concerns, resources required for misrepresenting information and falsifying numbers, and the stringency of fake news laws. With some abuse of language, I will hereafter interchangeably refer to $k$ as ``misreporting costs'' or ``costs' intensity.''\footnote{I use the quadratic loss $k(r-\theta)^2$ to obtain a closed-form solution and to simplify exposition. To find the equilibria of the communication subgame (Proposition~\ref{prop:monopoly} in Appendix~\ref{app:communication}), I use a more general cost function $kC(r,\theta)$.}

{\bf Influential News.} The media outlet is influential only if the voter's sequentially rational decision is not constant along the equilibrium path. To ensure that the outlet is influential, I assume that the state space is relatively large, i.e., $\phi\geq 3\gamma(\varphi_v-\varphi_m)^2$. Intuitively, a larger state space implies more uncertainty over quality and thus a more prominent role for an informed outlet. This assumption is sufficient to guarantee that in equilibrium the outlet is influential and that candidates cannot gain from proposing policies that make the outlet superfluous.\footnote{See Corollary~\ref{cor:support} in Appendix~\ref{app:policy}.} 

{\bf Strategies and Equilibrium.} A strategy for the incumbent is a binding commitment to a policy proposal $q_i\in\mathbb{R}$; a strategy for the challenger is a function $q_c:\mathbb{R}\to \mathbb{R}$ that assigns a policy $q_c\in \mathbb{R}$ to each incumbent's proposal $q_i$. I assume that candidates cannot condition their proposals on the state or on the outlet's reports.\footnote{This assumption is in line with the idea that all uncertainty about quality is publicly resolved only after policy implementation, and policies cannot be easily changed in the short run. Moreover, candidates cannot credibly and profitably condition their proposals on the media outlet's reports.} A reporting strategy for the media outlet is a function $\rho:\Theta\times \mathbb{R}^2\to\mathbb{R}$ that associates a news report $r\in\mathbb{R}$ to every tuple of proposals $q\in\mathbb{R}^2$ and quality $\theta\in\Theta$. I say that a report $r$ is off-path if, given strategy $\rho(\cdot)$, $r$ will not be observed by the voter. Otherwise, I say that $r$ is on-path. A belief function for the voter is a mapping $p:\mathbb{R}\times\mathbb{R}^2\to\Delta(\Theta)$ that, given any news report $r\in\mathbb{R}$ and policies $q\in\mathbb{R}^2$, generates posterior beliefs $p(\theta|r,q)$, where $p(\cdot)$ is a probability density function. Given a report $r$ and posterior beliefs $p(\theta|r,q)$, the voter casts a ballot for a candidate in the sequentially rational set $\beta(r,q)=\argmax_{b\in\{i,c\}}\mathbb{E}_p[u_v(b,\theta,q)\,|\, r]$.

The solution concept is the perfect Bayesian equilibrium (PBE), refined by the Intuitive Criterion \citep{cho1987}.\footnote{For a textbook definition of PBE and Intuitive Criterion, see \cite{fudenbergtirole1991}. A formal definition of Intuitive Criterion is also provided in Appendix~\ref{app:communication}} For most of the analysis, I focus on the sender-preferred equilibrium defined as follows: when the voter is indifferent between the two candidates at a given belief, she selects the one endorsed by the media outlet; when a candidate is indifferent between some proposals, she advances the policy closest to the outlet's bliss $\varphi_m$. Given the potential conflict of interest between the voter and the media outlet, the sender-preferred equilibrium is also the least preferred by the voter. The focus on this type of equilibrium provides a useful benchmark consisting of the voter's worst-case scenario, which is key for the robust-control approach to policy analysis \citep{hansen2008robustness}. Moreover, it selects the unique perfect sequential equilibrium \citep{grossman1986perfect} of the communication subgame. I hereafter refer to a sender-preferred PBE robust to the Intuitive Criterion simply as ``equilibrium.'' 

{\bf Discussion and Extensions.} Most qualitative findings are robust to the model's specifications.  Section~\ref{sec:extensions} looks at several variations of
the baseline model as robustness checks. These include the cases of a purely ideological media outlet and simultaneous policy-making, among others. There, I also discuss the suitability of the equilibrium concept, and various assumptions on the information and cost structure.

\section{Equilibrium}\label{sec:eqm}
I present the main equilibrium analysis in two parts: in Section~\ref{sec:commsub}, I begin by solving for the equilibrium of the final communication subgame where, given any fixed pair of policies, the media outlet delivers to the voter a news report about the candidates' quality. In Section~\ref{sec:policymaking} I proceed by studying the equilibrium of the policy-making stage, where candidates sequentially commit to policy proposals. Formal proofs are relegated to Appendix~\ref{app:communication} and Appendix~\ref{app:policystage}.

\subsection{The Communication Subgame}\label{sec:commsub}
The communication subgame takes place after both candidates commit to policy proposals. The media outlet privately observes the candidates' relative quality $\theta$ and then delivers a news report $r$ consisting of a literal statement about $\theta$. The voter, after observing the outlet's report but not the candidates' quality, casts a ballot for either the incumbent or the challenger. For convenience, I denote the communication subgame by $\hat\Gamma$.

Given proposals $q$, the media outlet has a conflict of interest with the voter when quality is between the thresholds $\tau_j(q)$, $j\in\{m,v\}$. Consider for example the case where policies $q$ are such that\footnote{We have that $\tau_m(q)>\tau_v(q)$ when $q_c<q_i$, $\tau_m(q)<\tau_v(q)$ when $q_c>q_i$, and $\tau_m(q)=\tau_v(q)$ when $q_c=q_i$. In the latter case, there is no conflict of interest between the media outlet and the voter.} $\tau_m(q)>\tau_v(q)$. When $\theta>\tau_m(q)$ (resp. $\theta<\tau_v(q)$), the voter and the outlet both agree that the best candidate is the incumbent (resp. the challenger). By contrast, when $\theta\in(\tau_v(q),\tau_m(q))$ the voter prefers the incumbent while the outlet endorses the challenger. Since the voter cannot observe the realized quality, she is uncertain about whether a conflict of interest exists or not. Figure~\ref{fig:conflict} illustrates the preferred candidate of the media outlet and the voter across different states and for the case $\tau_m(q)>\tau_v(q)$.

\begin{figure}
	\centering
	\includegraphics[width=\linewidth]{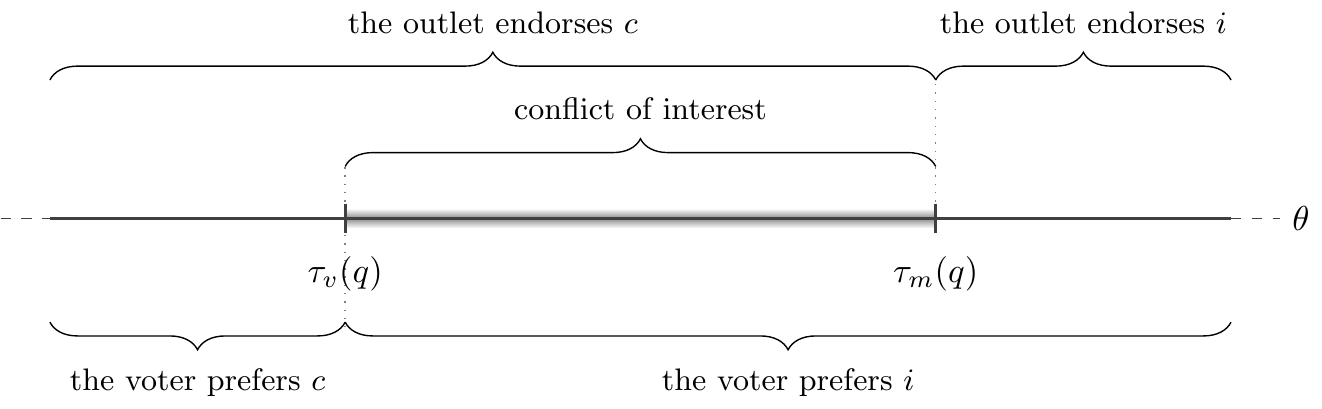}
	\caption{The media outlet and the voter's favorite candidate for different levels of quality. The states in which there is a conflict of interest are highlighted in gray.}
	\label{fig:conflict}
\end{figure}

The media outlet can misreport its private information about quality so as to induce the election of its endorsed candidate $\hat m(q,\theta)$ and seize the gains $\xi$. Denote by $\hat\Theta(q)$ the set of states that lie strictly between the thresholds $\tau_j(q)$, $j\in\{m,v\}$. If the media outlet delivers a report that induces the election of its endorsed candidate when there is a conflict of interest, then I say that persuasion has occurred.
\begin{definition}[Persuasion]
	The media outlet persuades the voter if $\beta(\rho(\theta,q),q)=\hat m(\theta,q)$ for some $\theta\in\hat\Theta(q)=\left(\min\left\{\tau_v(q),\tau_m(q)\right\}, \max\left\{\tau_v(q),\tau_m(q)\right\}\right)$.
\end{definition}

Since misreporting is costly, there is a limit to the reports that the outlet can profitably deliver in a certain state, and thus different reports each carry a different informational content that is not arbitrarily determined by the voter's strategic inference.\footnote{As it is the case, for example, in cheap talk games.} Consider a news report, $r>\tau_m(q)$, indicating that quality is sufficiently high for the outlet to endorse the incumbent. Suppose now that $r$ leads to the electoral victory of the outlet's endorsed candidate, $\beta(r,q)=i$. I define the ``lowest misreporting type'' $l(r)$ as the highest state\footnote{In the jargon commonly used in signaling games, the state is also referred to as the ``type'' of sender.} $\theta$ in which the outlet does not find it strictly profitable to deliver the news report $r$. More formally, for some report $r>\tau_m(q)$ such that $\beta(r,q)=i$,
\[
l(r)=\max\left\{ r-\sqrt{\frac{\xi}{k}},\tau_m(q) \right\}.
\]
In equilibrium, the voter understands that such a report $r$ could not be profitably delivered if quality is lower than $l(r)$, and should accordingly place probability zero on every $\theta<l(r)$. I similarly define  the ``highest misreporting type'' $h(r)$ as the lowest state in which the outlet does not find it strictly profitable to deliver a news report $r<\tau_m(q)$ such that $\beta(r,q)=c$. Formally,
\[
h(r)=\min\left\{ r+\sqrt{\frac{\xi}{k}},\tau_m(q) \right\}.
\]

I can now present the main result of this section: in the equilibrium of the communication subgame $\hat\Gamma$, the media outlet ``jams'' information by delivering the same pooling report $r^*(q)$ whenever quality takes values around the voter's threshold $\tau_v(q)$. Otherwise, when quality is relatively far from $\tau_v(q)$, the outlet always reports truthfully. When observing the pooling report $r^*(q)$, the voter 's expectation about quality is exactly $\tau_v(q)$, and therefore she is indifferent between the two candidates.\footnote{In the sender-preferred equilibrium, the voter selects the candidate endorsed by the media outlet when indifferent. Therefore, the voter never mixes.} This result helps us to find the candidates' equilibrium probability of electoral victory given any pair of proposals $q$.

\begin{lemma}\label{lemma:senderpm}
	The equilibrium of the communication subgame $\hat\Gamma$ is a pair $(\rho(\theta,q), p(\theta|r,q))$ such that, given policy proposals\footnote{Up to changes of measure zero in $\rho(\theta,q)$ due to the media outlet being indifferent between reporting $l(r^*(q))$ and $r^*(q)$ (resp. $h(r^*(q))$ and $r^*(q)$) when the state is $\theta=l(r^*(q))>\tau_m(q)$ and $\tau_m(q)<\tau_v(q)$ (resp. $\theta=h(r^*(q))<\tau_m(q)$ and $\tau_m(q)>\tau_v(q)$).} $q$,
	\begin{enumerate}
		\item[i)] If $\tau_v(q)<\tau_m(q)$, then
		\[
		\rho(\theta,q) =
		\begin{cases}
		r^*(q)=\max\left\{\tau_v(q)-\tfrac{1}{2}\sqrt{\frac{\xi}{k}},2\tau_v(q)-\tau_m(q)\right\} & \quad \text{if } \; \theta \in \left(r^*(q),h(r^*(q))\right)\\
		\theta & \quad \text{otherwise. }
		\end{cases}
		\]
		\item[ii)] If $\tau_v(q)>\tau_m(q)$, then
		\[
		\rho(\theta,q) =
		\begin{cases}
		r^*(q)=\min\left\{\tau_v(q)+\tfrac{1}{2}\sqrt{\frac{\xi}{k}},2\tau_v(q)-\tau_m(q)\right\} & \quad \text{if } \; \theta \in \left(l\left(r^*(q)\right),r^*(q)\right)\\
		\theta & \quad \text{otherwise. }
		\end{cases}
		\]
		\item[iii)] If $\tau_v(q)=\tau_m(q)$, then $\rho(\theta,q)=\theta$ for all $\theta\in\Theta$.
		\item[iv)] Posterior beliefs $p(\theta\,|\,r,q)$ are according to Bayes' rule whenever possible and such that $\mathbb{E}_p [\theta \,|r^*(q)]=\tau_v(q)$, $\mathbb{E}_p [\theta \,|r]<\tau_v(q)$ for every off-path $r$, and $p(\theta|r,q)$ is degenerate at $\theta=r$ otherwise.
	\end{enumerate}
\end{lemma}
To understand the intuition behind Lemma~\ref{lemma:senderpm}, consider the case where proposals $q$ are such that $\tau_v(q)<\tau_m(q)$, and suppose that there exists a fully revealing equilibrium in truthful strategies, where $\rho(\theta,q)=\theta$ for every $\theta\in\Theta$. When quality is slightly higher than the voter's threshold $\tau_v(q)$, the media outlet can deliver some report $r\leq \tau_v(q)$ such that the incurred misreporting costs are lower than the gains obtained from endorsing the winning candidate, i.e., $k(r-\theta)^2<\xi$. Given the truthful reporting rule $\rho(\theta,q)$, the voter takes the outlet's reports at face value, and thus casts a ballot for the challenger after observing any $r\leq\tau_v(q)$. The outlet has a strictly profitable deviation, implying that in equilibrium there must be misreporting in some state.

Misreporting is a costly activity, and therefore the media outlet misreports only if doing so induces the electoral victory of its endorsed candidate $\hat m(\theta,q)$. Moreover, if it is profitable for the outlet to deliver a report $r'<\tau_m(q)$ when quality is $\theta'\in(r',\tau_m(q))$, then reporting $r'$ must be profitable for all $\theta\in[r',\theta']$. This suggests that in equilibrium the outlet ``pools'' information about quality by delivering the same report $r^*(q)$ for different states in a convex set $S(r^*(q))$ such that $\hat m(\theta',q)=\hat m(\theta'',q)$ for all $\theta',\theta''\in S(r^*(q))$.

Upon observing the pooling report $r^*(q)$, the voter infers that the realized quality is in the set $S(r^*(q))$. If the voter's expectation about quality $\mathbb{E}_p[\theta|r^*(q)]$ is greater than her threshold $\tau_v(q)$, then she casts a ballot for the incumbent; otherwise she votes for the challenger. Therefore, by pooling states around $\tau_v(q)$ in a way such that $\mathbb{E}_p[\theta|r^*(q)]\leq\tau_v(q)$, the outlet can induce the election of the challenger even when quality is such that the voter's preferred candidate is the incumbent. That is, the outlet can achieve persuasion by pooling information about quality.

The candidate endorsed by the media outlet is more likely to be elected when the pooling report $r^*(q)$ makes the voter just indifferent between casting a ballot for the incumbent and the challenger: pooling reports that induce lower expectations have the same effect on the voter's choice but are more expensive to deliver when there is a conflict of interest. Therefore, in equilibrium the outlet misreports by delivering a pooling report $r^*(q)$ that jams states around the voter's threshold in a way such that $\mathbb{E}_p[\theta|r^*(q)]=\tau_v(q)$.

This kind of pooling strategy prescribes the outlet to misreport even in states where there is no conflict of interest. Even though at first it might seem counterintuitive, this reporting behavior is consistent with strategic skepticism: the voter, being aware of the media outlet's leaning and misreporting technology, demands sufficiently strong evidence that quality is low enough for the challenger to be elected. Therefore, when quality is just slightly below $\tau_v(q)$, the outlet must nevertheless misreport to overcome the voter's skepticism. 

By contrast, truthful reporting always occurs when quality takes extreme values that are relatively far from the voter's threshold $\tau_v(q)$. There are two possibilities in this case: either there is a conflict of interest, or the interests of the outlet and the voter are aligned. In the former case, misreporting is not convenient for the outlet as it would be prohibitively costly to deliver a report that induces the election of its endorsed candidate. In the latter case, the outlet does not need to misreport because the true realization of quality is a sufficiently discriminating signal for the voter be trustful.  

Lemma~\ref{lemma:senderpm} shows that, given proposals $q$, the media outlet persuades the voter when $\theta\in(\tau_v(q),h(r^*(q)))$ if $\tau_v(q)<\tau_m(q)$ and when $\theta\in(l(r^*(q)),\tau_v(q))$ if $\tau_v(q)>\tau_m(q)$. If $\tau_v(q)=\tau_m(q)$, then there cannot be persuasion since the outlet and the voter always agree on which candidate is best. By contrast, I say that the outlet exerts ``full persuasion'' if persuasion occurs in every state in which there is a conflict of interest.
\begin{definition}[Full persuasion]\label{def:fullpers}
	The media outlet exerts full persuasion if $\beta(\rho(\theta,q),q)=\hat m(\theta,q)$ for all $\theta\in\hat\Theta(q)$.
\end{definition}
The media outlet has fully persuasive power if, given policy proposals $q$, the misreporting costs $k$ are low enough to make persuasion affordable in every state where there is a conflict of interest. Formally, there is full persuasion if $k\in\left(0,\hat k(q)\right]$, where\footnote{The cost threshold $\hat k(q)$ is obtained by setting $h(r^*(q))=\tau_m(q)$ for $\tau_v(q)<\tau_m(q)$ or $l(r^*(q))=\tau_m(q)$ for $\tau_v(q)>\tau_m(q)$, where $r^*(q)$ is defined as in Lemma~\ref{lemma:senderpm}.}
	\[
	\hat k(q)= \frac{\xi}{4\gamma^2 (\tau_v(q)-\tau_m(q))^2}.
	\]
Alternatively, the outlet exerts full persuasion if, given misreporting costs $k$, the proposals $q_i$ and $q_c$ are sufficiently close to each other. Intuitively, as candidates' policies become more similar, the preferences of the voter and the outlet become more aligned, and the set of states in which there is a conflict of interest becomes smaller. Since the outlet's potential gains $\xi$ are fixed, the share of states in which persuasion occurs under a conflict of interest increases as proposals get closer. If policies are sufficiently similar, then persuasion occurs every time there is a conflict of interest. Formally, there is full persuasion when proposals $q$ are such that 
	\begin{equation}\label{eq:fullpers}
	(q_c-q_i)^2\leq\frac{\xi}{16\gamma^2 (\varphi_m-\varphi_v)^2 k}.
	\end{equation}

Figure~\ref{fig:monopoly} shows the equilibrium reporting rule of Lemma~\ref{lemma:senderpm} for different policies and misreporting costs. In panel (a), the outlet is more likely than the voter to prefer the challenger and misreporting costs are relatively high. In this case, the media outlet discredits the incumbent by delivering a report that ``belittles'' realizations of quality around the voter's threshold $\tau_v(q)$. With this strategy, the outlet achieves persuasion in those states that are highlighted in light gray. By contrast, truthful reporting occurs despite a conflict of interest in states that are highlighted in dark gray: in these cases, persuasion is prohibitively expensive because of the relatively high misreporting costs $k>\hat k(q)$. In states that are highlighted in gray, the outlet spends resources to misreport information even though no conflict of interest is in place. These ``white lies'' are the result of the voter's skepticism about news reports that are not sufficiently discriminatory. 

Panel (b) of Figure~\ref{fig:monopoly} shows the equilibrium reporting rule when the outlet is more likely than the voter to prefer the incumbent and misreporting costs are relatively low. In this case, the media outlet supports the incumbent by delivering reports that ``exaggerate'' realizations of quality\footnote{There are perfect Bayesian equilibria in the communication subgame $\hat\Gamma$ where the outlet supports the incumbent (resp. challenger) by delivering a report that is lower (resp. higher) than the actual realization of quality. These equilibria do not survive the Intuitive Criterion test.} around the voter's threshold $\tau_v(q)$. Low misreporting costs allow the outlet to exert full persuasion and induce the election of its endorsed candidate every time there is a conflict of interest. As before, states in which the outlet delivers white lies are highlighted in gray, while states in which the outlet persuades the voter are in light gray.
\begin{figure}[]
	\centering
	\begin{subfigure}[b]{0.48\textwidth}
		\includegraphics[width=\textwidth]{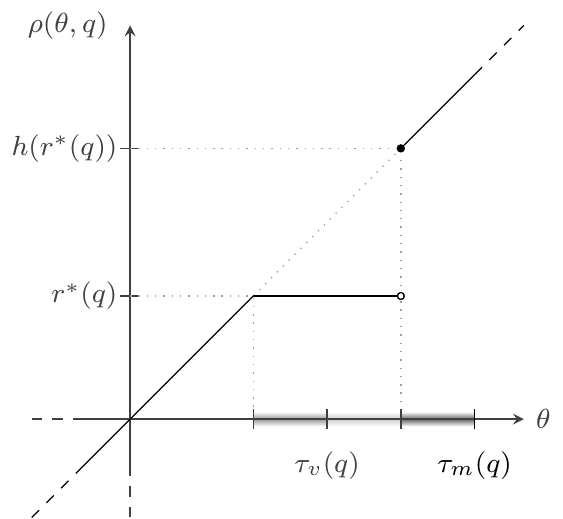}
		\caption{$\tau_v(q)<\tau_m(q)$ and $k>\hat k(q)$.}
	\end{subfigure}
	\begin{subfigure}[b]{0.48\textwidth}
		\includegraphics[width=\textwidth]{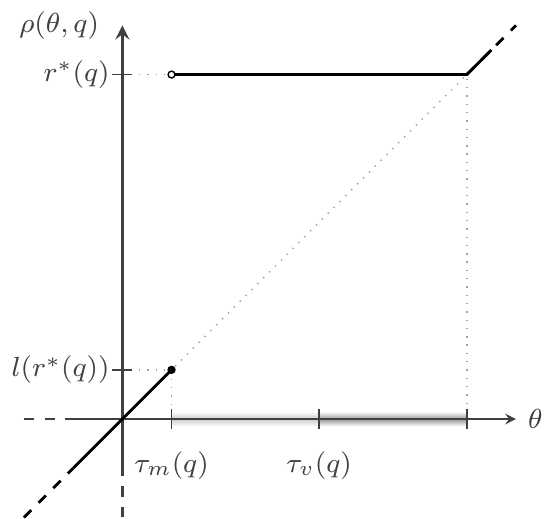}
		\caption{$\tau_v(q)>\tau_m(q)$ and $0<k\leq \hat k(q)$.}
	\end{subfigure}
	\caption{The two panels illustrate the equilibrium reporting rule for different levels of misreporting costs and ordering of policy proposals. The states in which persuasion occurs are highlighted in light gray. The states in which the outlet misreports even though there is no conflict of interest are highlighted in gray. The states in which the outlet reveals the true realization of quality even though there is a conflict of interest are highlighted in dark gray.}
	\label{fig:monopoly}
\end{figure}

\subsection{The Policy-making Stage}\label{sec:policymaking}
Consider now the policy-making stage, where candidates sequentially make a binding commitment to a policy proposal. Since candidates are purely office-seeking, they advance policies to maximize their chances of getting elected. The result in the previous section is key for finding the candidates' equilibrium proposals: Lemma~\ref{lemma:senderpm} shows the media outlet's equilibrium reporting rule and thus pins down the candidates' probability of electoral victory given any pair of policies $q$.

I denote by $q_i^*(k)$ the equilibrium policy advanced by the incumbent and by $q^*_c(q_i,k)$ the challenger's best response to some proposal $q_i$. I refer to policies that are relatively close to the voter's bliss $\varphi_v$ as ``populist'' and to policies that are relatively close to the outlet's bliss $\varphi_m$ as ``biased.'' The next result establishes the equilibrium proposals $q^*(k)=(q_i^*(k),	q_c^*(q_i^*(k),k))$ as a function of the misreporting costs' intensity $k$.
\begin{proposition}\label{prop:eqmpol}
	The equilibrium policies $q^*(k)$ are
	\[
	q_i^*(k)=
	\begin{cases}
	\varphi_v + \frac{\sqrt{\xi/k}}{4\gamma(\varphi_v-\varphi_m)} - \sqrt[4]{\frac{\xi}{\gamma^2k}} & \text{ if } k>\bar k\\
	\frac{\varphi_v+\varphi_m}{2}-\frac{\sqrt{\xi/k}}{4\gamma(\varphi_v-\varphi_m)} & \text{ if } k\in\left(\bar k/4,\bar k\right]\\
	\varphi_m & \text{ if } k\in\left(0,\bar k/4\right]
	\end{cases}
	\]
	
	\[
	q_c^*(q_i^*(k),k)=
	\begin{cases}
	\varphi_v - \sqrt[4]{\frac{\xi}{\gamma^2k}} & \text{ if } k>\bar k\\
	\varphi_m & \text{ if } k\in\left(0,\bar k\right]
	\end{cases}
	\]
	where the misreporting costs threshold is $\bar k = \frac{\xi}{\gamma^2 (\varphi_v-\varphi_m)^4}$.
\end{proposition}

In equilibrium, proposals are weakly increasing in $k$ and strictly increasing for every finite\footnote{The threshold $\bar k$ is the highest intensity of misreporting costs such that the challenger best responds with $\varphi_m$ when undercutting the incumbent's proposal.} $k>\bar k=\frac{\xi}{\gamma^2 (\varphi_v-\varphi_m)^4}$. When the costs of misreporting information are relatively low, i.e. for $k\in\left(0,\bar k/4\right]$, both candidates advance the media outlet's bliss policy $\varphi_m$. Thus, variations of $k$ within the region $\left(0,\bar k/4\right]$ leave the equilibrium policies unaltered. For intermediate costs, $k\in\left(\bar k/4,\bar k\right]$, the incumbent proposes increasingly moderate policies $q_i^*(k)\in\left(\varphi_m,\frac{\varphi_v+\varphi_m}{2}\right)$, while the challenger keeps offering the outlet's bliss $\varphi_m$. Thus, an increase of $k$ in the region $\left(\bar k /4,\bar k\right)$ generates policy divergence. When the costs of misreporting are relatively high, $k>\bar k$, the challenger proposes less biased policies as well. As $k$ grows arbitrarily large, both proposals converge to the voter's preferred policy $\varphi_v$, with $q_i^*(k)>q_c^*(q_i^*(k),k)$ for every finite\footnote{Equilibrium proposals are continuous in $k$ as $\lim_{k\to \bar k^+} q_i^*(k)=\lim_{k\to \bar k^-} q_i^*(k)$ and $\lim_{k\to \bar k/4} q_i^*(k)=\lim_{k\to \bar k/4} q_c^*(q_i^*(k),k)=\varphi_m$.} $k>\bar k/4$. Figure~\ref{fig:policies} illustrates the equilibrium policy proposals for different levels of misreporting costs' intensity $k$.

\begin{figure}
	\centering
	\includegraphics[width=0.65\linewidth]{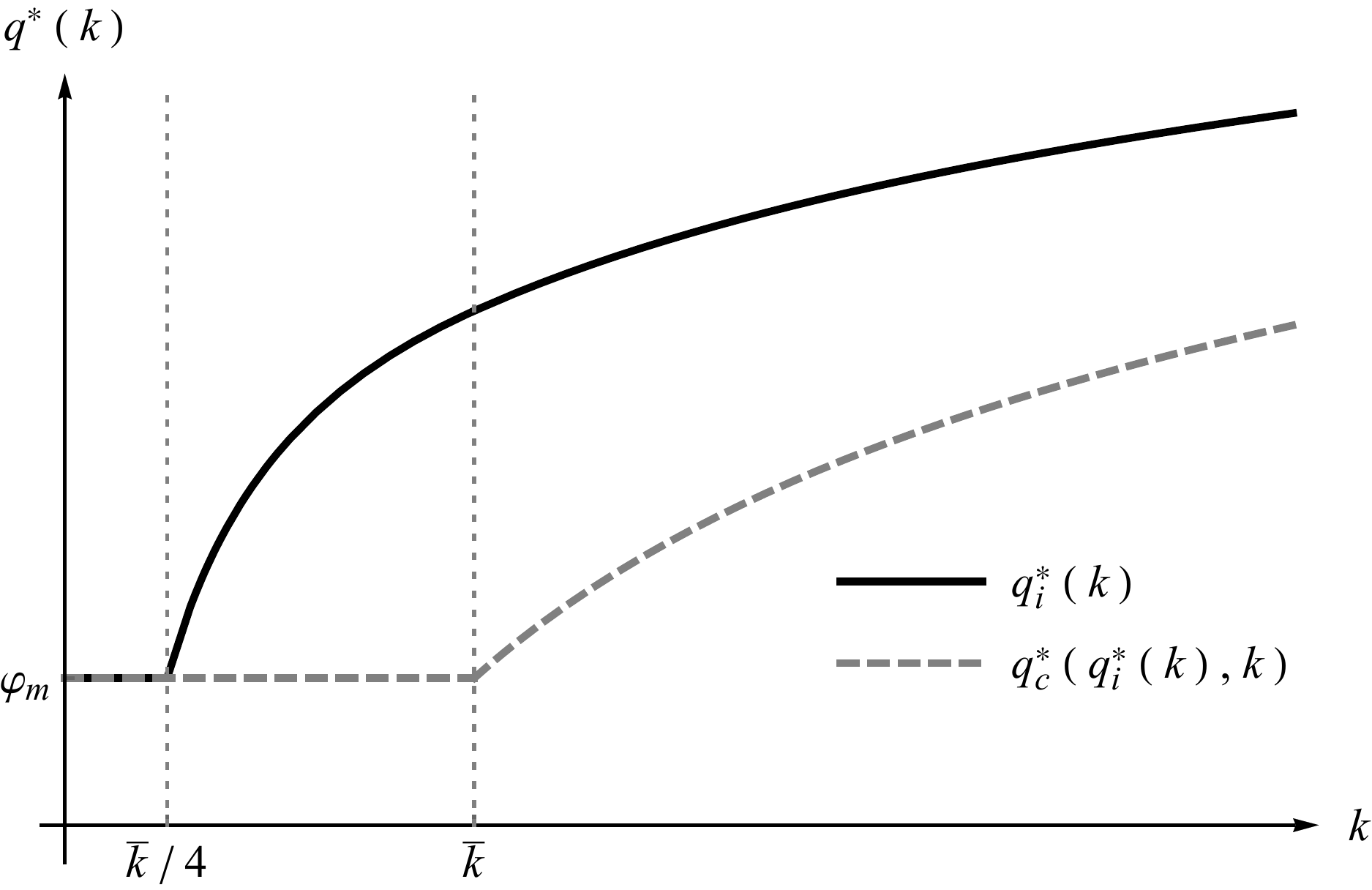}
	\caption{Equilibrium policy proposals for different intensities of misreporting costs. As $k$ grows arbitrarily large, both proposals monotonically converge to $\varphi_v$.}
	\label{fig:policies}
\end{figure}

Here I discuss the intuition behind Proposition~\ref{prop:eqmpol}. Since policies outside the set $[\varphi_m,\varphi_v]$ are always dominated, I restrict attention\footnote{The focus on policies in the set $[\varphi_m,\varphi_v]$ is without loss of generality: for both candidates $j\in\{i,c\}$, proposals $q_j>\varphi_v$ (resp. $q_j<\varphi_m$) are dominated by every $q_j'\in[\varphi_v,q_j)$ (resp. $q_j'\in(q_j,\varphi_m]$) as any such $q_j'$ is more appealing to both the voter and the outlet.} to $q_j\in[\varphi_m,\varphi_v]$, $j\in\{i,c\}$. First, consider the challenger's problem of best responding to the incumbent's proposal.  When the incumbent proposes a relatively populist policy, the challenger's best response is to ``undercut'' the incumbent with the most biased policy $q_c<q_i$ that endows the media outlet with fully persuasive power.\footnote{Formally, the challenger offers the lowest proposal that satisfies condition~\eqref{eq:fullpers} with equality.} With this strategy, the challenger maximizes both the extent of the conflict of interest $\left(\tau_m(q),\tau_v(q)\right)$ and the probability of receiving the outlet's support, subject to the outlet exerting full persuasion. Even though the challenger's best response is less appealing to the voter, the loss in ``popular appeal'' is more than compensated by the outlet's ability to persuade the voter over a large set of contingencies. By contrast, offering a more populist policy $q_c>q_i$ would make the challenger slightly more appealing to the voter at the expense of getting the incumbent into the good graces of a fully persuasive media outlet. Thus, offering any $q_c>q_i$ is suboptimal in this case.

When the incumbent proposes relatively biased policies and misreporting costs are sufficiently high, the best response of the challenger is to offer the voter's bliss $\varphi_v$. This strategy generates a large conflict of interest $\hat\Theta(q)=\left(\tau_m(q),\tau_v(q)\right)$ such that the voter requires evidence that quality is exceptionally high in order to vote for the incumbent. The outlet is now more likely to endorse the incumbent than the challenger but, because of high misreporting costs and a sharp policy divergence, it cannot exert full persuasion. In this case, offering the voter's bliss $\varphi_v$ is the challenger's best response because it leaves the incumbent with an unpopular policy and the support of a weakened media outlet. 

Similarly, when the incumbent's policy is relatively biased but misreporting costs are sufficiently low, the challenger's best response remains that of undercutting the incumbent. The strategy of proposing the voter's bliss now backfires because with a low intensity of misreporting costs the media outlet retains its ability to persuade the voter in a relatively large share of $\hat\Theta(q)$. If the costs' intensity $k$ is low enough, then the challenger's best response is to undercut the proposal of the incumbent to the point of offering the outlet's bliss $\varphi_m$. Figure~\ref{fig:bestresp} shows the challenger's best response for different intensities of misreporting costs.\footnote{Proposition~\ref{prop:bestresp} in Appendix~\ref{app:best} shows the challenger's best response function.}

\begin{figure}
	\centering
	\includegraphics[width=0.5\linewidth]{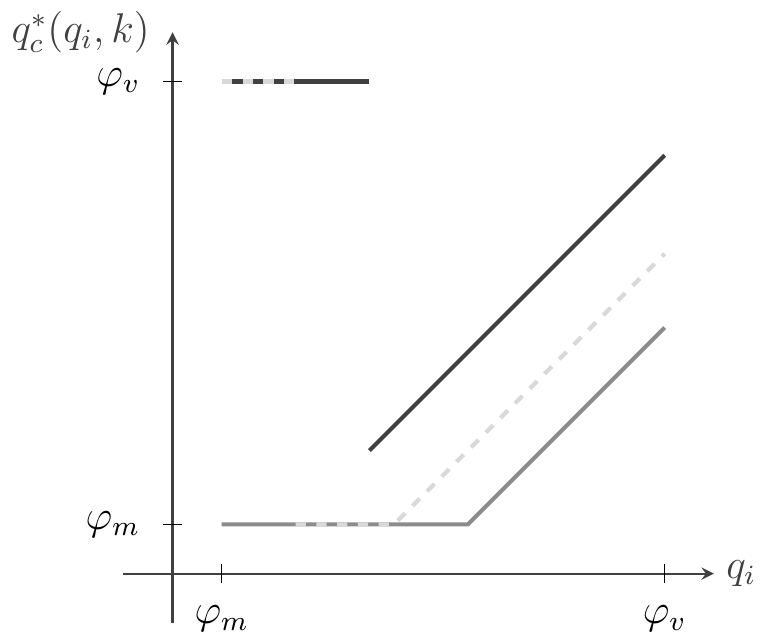}
	\caption{The challenger's best response for different intensities of misreporting costs. The best response is depicted in black for relatively high costs, $k>\bar k$; in dashed light gray for intermediate costs, $k\in\left(\bar k/4,\bar k\right)$; in dark gray, for relatively low costs, $k\in\left(0,\bar k /4\right]$.}
	\label{fig:bestresp}
\end{figure}

Consider now the incumbent's problem of selecting a policy that maximizes her probability of electoral victory, and suppose first that the intensity of misreporting costs is relatively high, $k>\bar k/4$. In this case, the optimal proposal of the incumbent $q_i^*(k)$ is the policy that makes the challenger indifferent whether to best reply with the voter's bliss or with a relatively more biased policy (i.e., by undercutting the incumbent): higher proposals $q_i>q_i^*(k)$ would allow the challenger to get the support of a fully persuasive outlet; lower proposals $q_i<q_i^*(k)$ would be highly unpopular in comparison with the challenger's best response of offering the voter's bliss. By contrast, when the intensity of misreporting costs is relatively low, $k\in\left(0,\bar k/4\right)$, the media outlet exerts full persuasion for any combination of candidates' proposals $q\in[\varphi_m,\varphi_v]^2$. In this case, the incumbent's optimal policy is to offer the outlet's bliss, as any higher proposal $q_i>\varphi_m$ would allow the challenger to get the support of a fully persuasive media outlet by undercutting\footnote{Notice that the model does not predict that the incumbent always takes a more populist position than to the challenger. While this happens in the sender-preferred equilibrium, there are other equilibria where the challenger goes fully populist by offering the voter's favorite policy and the incumbent proposes a relatively more biased policy. See Proposition~\ref{prop:bestresp}.} $q_i$.

The presence of a persuasive media outlet generates a distortion in the process of policy-making. Since candidates look to gain both the consensus of the voter and the support of the influential outlet, their proposals drift away from the voter's preferred policy, breaking down the centripetal force of the median voter theorem \citep{downs1957,black1948}. This distortion peaks when the intensity of misreporting costs is sufficiently low so that both candidates advance the media outlet's bliss policy. In this case, persuasion never takes place since there is no conflict of interest when candidates' proposals are identical (see Lemma~\ref{lemma:senderpm}). Therefore, with lower (resp. higher) intensities of misreporting costs the voter might have worse (resp. better) policies but more (resp. less) information about quality. In the next section, I study this trade-off in relation to the voter's welfare.

\section{Welfare and Regulation}\label{sec:welfare}
Having characterized the equilibrium of the communication subgame (Lemma~\ref{lemma:senderpm} in Section~\ref{sec:commsub}) and the candidates' equilibrium proposals (Proposition~\ref{prop:eqmpol} in Section~\ref{sec:policymaking}), I now proceed to study their welfare implications. I denote by $W_v^*(k)$ the voter's equilibrium expected utility and refer to $W_v^*(k)$ simply as the voter's welfare.\footnote{See equation~\eqref{eq:welfare} in Appendix~\ref{app:welfare} for an explicit formulation of $W^*_v(k)$.} As a benchmark, consider the voter's expected utility under complete information, which I denote by $\hat W_v$. Suppose that the voter perfectly observes the realized quality after the policy-making stage but before the election takes place. In this case, both candidates cannot do better than offer the voter's bliss policy as the media outlet would have no role. The candidate with the highest relative quality is always elected and the voter's favorite policy is always implemented. Therefore, $\hat W_v=\phi/4$. 

In Section~\ref{sec:voterwelfare}, I study how the intensity of misreporting costs affects different determinants of the voter's welfare. In Section~\ref{sec:incumbent}, I extend the main model by allowing candidates to select the costs' intensity ahead of the policy-making stage. Formal proofs are relegated to Appendix~\ref{app:welfare}.

\subsection{The Voter's Welfare}\label{sec:voterwelfare}

Consider the problem of a regulator that seeks to maximize the welfare of the voter by selecting the intensity of misreporting costs. This type of intervention can be performed, for example, by issuing ``fake news laws'' or by subsidizing watchdogs that expose to the public those media outlets that concoct news reports. As we have seen in the previous section, the process of policy-making is strategically intertwined with the voter's informational environment. Interventions that change the misreporting costs might affect both the amount of information received by the voter and the policies advanced by the candidates.\footnote{The voter receives more (resp. less) information if, given the outlet's reporting rule $\rho(\cdot)$, she casts a ballot for her preferred candidate in more (resp. less) states.} To maximize the voter's welfare, it is crucial for regulators to understand the consequences and the trade-offs involved in these types of interventions.

Before showing the next result, it is thus useful to note some important features of the equilibria in Proposition~\ref{prop:eqmpol} and Lemma~\ref{lemma:senderpm}.
First, equilibrium policies $q^*(k)$ satisfy condition~\eqref{eq:fullpers} for every finite $k$: on the equilibrium path, the media outlet always exerts full persuasion and the candidate endorsed by the outlet, $\hat m(q^*(k),\theta)$, is always elected. Figure~\ref{fig:monopoly3} shows the outlet's reporting rule on the equilibrium path for some finite $k>\bar k/4$. Second, an increase in the misreporting costs' intensity does not necessarily yield more information to the voter. Intuitively, since the outlet exerts full persuasion, the larger the conflict of interest $\hat\Theta(q^*(k))$, the less the information received by the voter. Recall that the share of states in which there is a conflict of interest is directly proportional to the difference between proposals.\footnote{Formally, $|\hat\Theta\left(q^*(k)\right)|=2\gamma(\varphi_v-\varphi_m)\left(q_i^*(k)-q_c^*\left(q_i^*(k),k\right)\right)$.} It follows that an increase in $k$ yields more (resp. less) information to the voter only if it generates policy convergence (resp. divergence). However, Proposition~\ref{prop:eqmpol} shows that the distance between equilibrium proposals is non-monotonic\footnote{A marginal increment in $k$ generates policy convergence for all finite $k>\bar k$, policy divergence for all $k\in\left(\bar k/4, \bar k\right)$, and has no effect on policies for all $k\in\left(0,\bar k/4\right)$.} in $k$. Hence, an increase in the misreporting costs' intensity might decrease the amount of information received by the voter in equilibrium. Third, with a relatively low intensity of misreporting costs, $k\in\left(0,\bar k/4\right)$, there is no conflict of interest because equilibrium policies are identical. In this case, persuasion never occurs and the media outlet fully reveals its private information about quality. By contrast, persuasion always takes place with positive probability for every finite $k>\bar k/4$. Therefore, in equilibrium there is persuasion only if the misreporting costs' intensity is sufficiently high.

\begin{figure}
	\centering
	\minipage{0.5\textwidth}%
	\includegraphics[width=\linewidth]{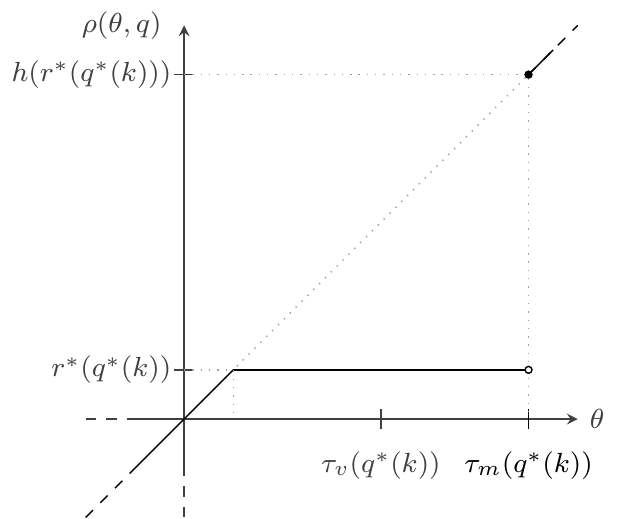}
	\endminipage
	\caption{The media outlet's reporting rule on the equilibrium path for some finite costs' intensity $k>\bar k/4$. The outlet exerts full persuasion and its endorsed candidate is always elected.}\label{fig:monopoly3}
\end{figure}

Since an increase in the misreporting costs might yield the voter better policies at the expense of selecting the best candidate with lower probability, it is not clear how this type of intervention would affect the voter's welfare. The next proposition clears up this ambiguity by showing that increments in the costs' intensity $k$  never harm the voter.
\begin{proposition}\label{prop:welfare}
	The voter's equilibrium welfare $W_v^*(k)$ is independent of $k$ for all $k\in\left(0,\bar k/4\right)$, and strictly increasing in $k$ for all finite $k\geq \bar k/4$. As $k\to\infty$, $W_v^*(k)\to \hat W_v$.
\end{proposition}
Proposition~\ref{prop:welfare} shows that even in those cases where an increase in $k$ generates more persuasion, the gain that the voter obtains from having better policies always overcomes the expected loss in quality due to worse selection. Denote by $\chi (k)$ the ex-ante probability that persuasion occurs, called the ``persuasion rate.'' From Proposition~\ref{prop:eqmpol} and Lemma~\ref{lemma:senderpm}, we have that on the equilibrium path the persuasion rate is $\chi (k)=\frac{\tau_m(q^*(k))-\tau_v(q^*(k))}{2\phi}$. As observed before, the media outlet is more likely to persuade the voter when the set of states in which there is a conflict of interest $\hat \Theta(q^*(k))$ is larger. Figure~\ref{fig:welfpers} shows both the voter's welfare and the probability that persuasion occurs as a function of $k$. 

The policy/information trade-off occurs when $k\in\left[\bar k/4, \bar k\right)$: in this case, a marginal increase in the costs' intensity $k$ generates policy divergence, more disagreement, and thus a higher persuasion rate $\chi(k)$. As a consequence, the voter becomes increasingly likely to cast a ballot for the wrong candidate. The expected loss in quality due to worse selection is more than compensated by the availability of an increasingly populist policy advanced by the incumbent: since the proposals of both candidates are heavily skewed toward the media outlet's preferred policy, the voter obtains an exceptionally high gain from policies that, on average, are closer to her bliss. When $k\geq \bar k$, an increase in the costs' intensity generates proposals that are both more populist and closer to each other. The resulting policy convergence reduces the conflict of interest, and thus the persuasion rate $\chi(k)$ declines. In this case, the welfare of the voter increases because she obtains better policies and makes a better selection of candidates. By contrast, when $k\in\left(0,\bar k/4\right)$, a marginal increase in the costs' intensity has no effect on equilibrium policies and thus does not impact the voter's welfare either. As a result, lenient measures can actually decrease the voter's welfare when taking into consideration the public resources required to carry out the interventions.

\begin{figure}[]
	\centering
	\begin{subfigure}[b]{0.48\textwidth}
		\includegraphics[width=\textwidth]{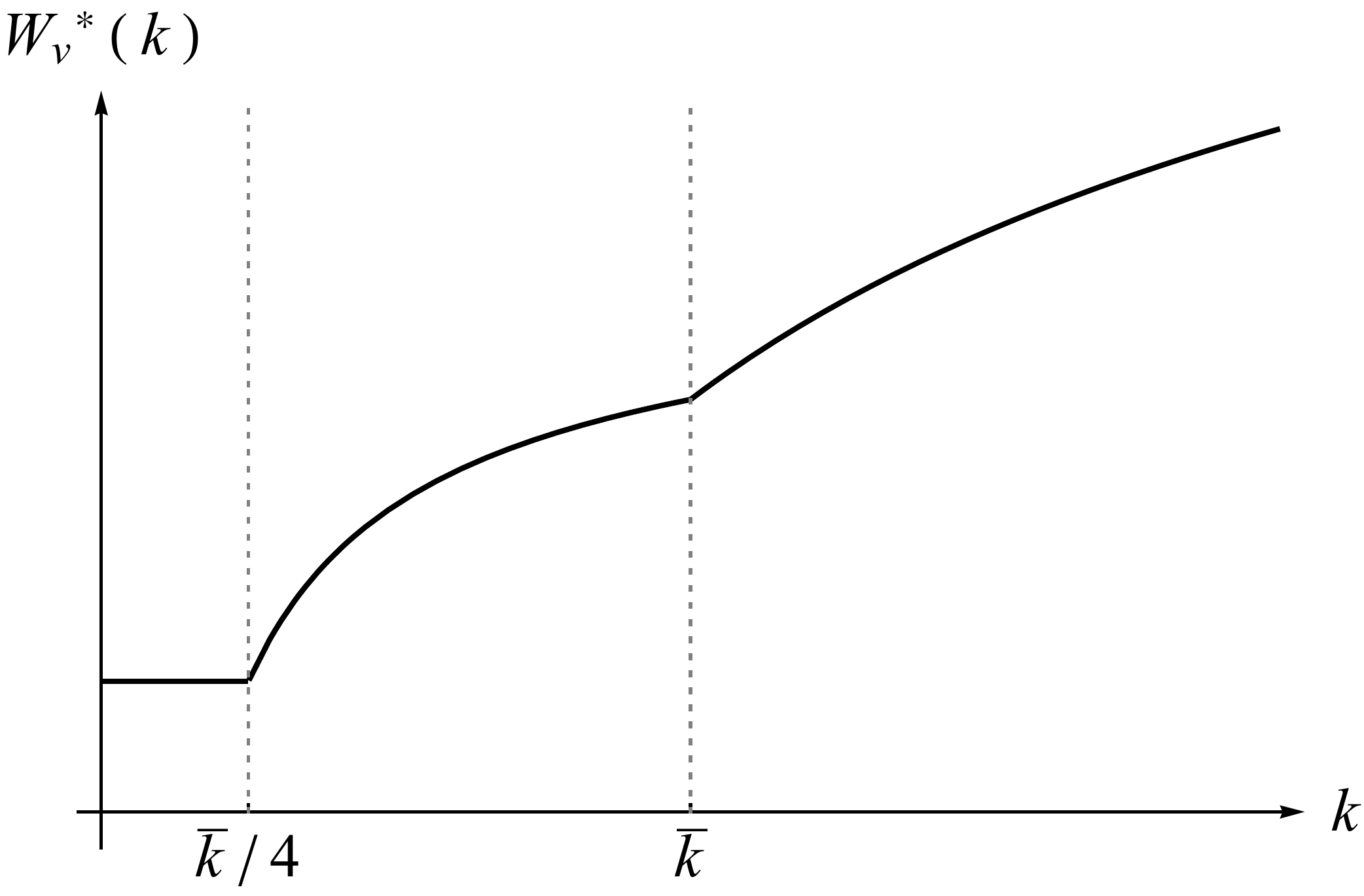}
		\caption{The voter's equilibrium welfare.}
		\label{rfidtest_xaxis}
	\end{subfigure}
	\begin{subfigure}[b]{0.48\textwidth}
		\includegraphics[width=\textwidth]{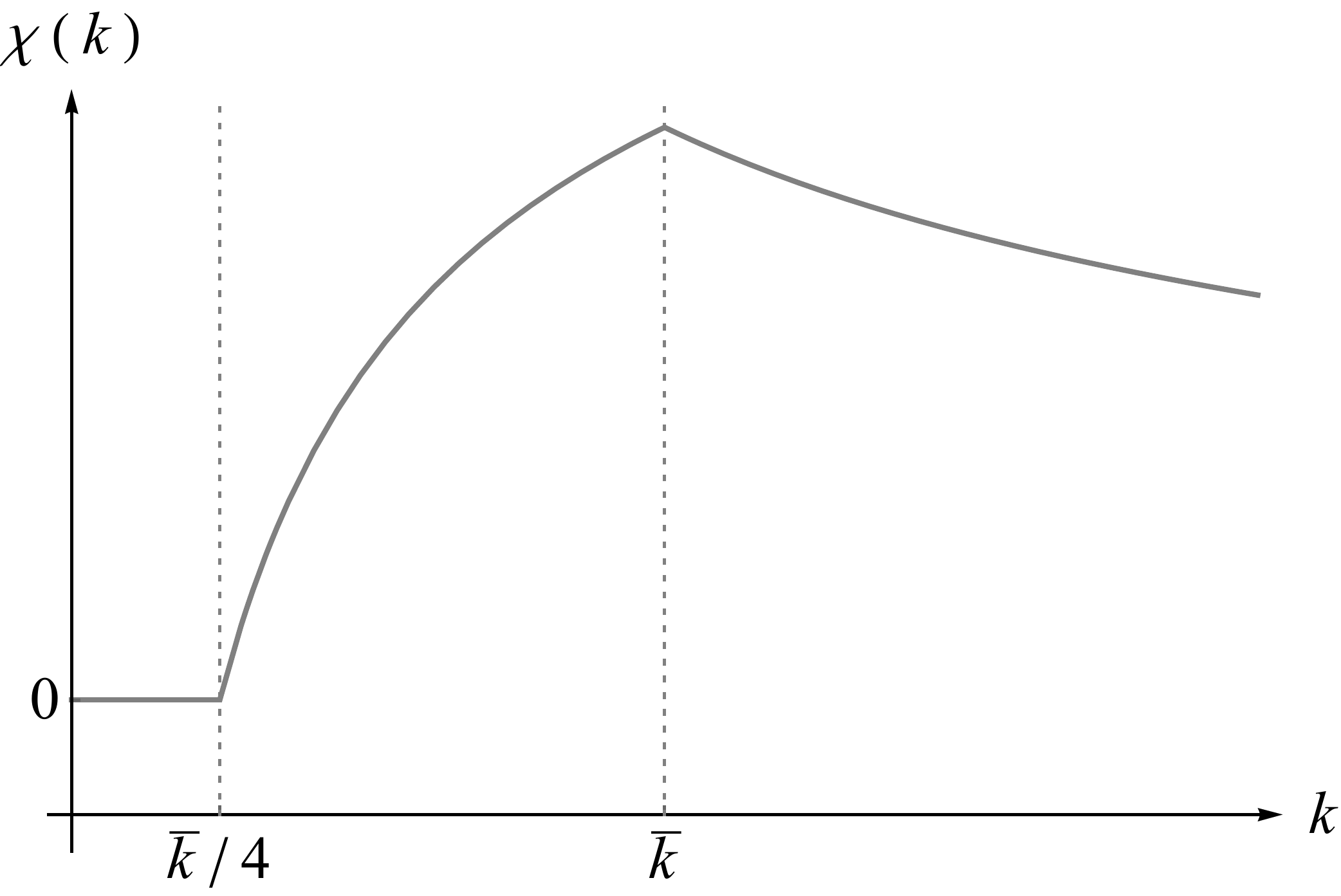}
		\caption{The probability of persuasion in equilibrium.}
		\label{rfidtest_yaxis}
	\end{subfigure}
	\caption{The voter's welfare increases with the intensity of misreporting costs even when higher costs' intensities leads to a higher persuasion rate. With relatively low intensities of misreporting costs there is no persuasion but the voter's welfare is at its minimum. As $k$ grows arbitrarily large, $W_v^*(k)$ and $\chi(k)$ converge monotonically to $\hat W_v=\phi/4$ and zero, respectively.}\label{fig:welfpers}
\end{figure}

The media outlet provides the voter with useful information about candidates' quality, but on the other hand it generates a policy distortion where proposals drift away from the voter's bliss. This trade-off reaches its peak when the intensity of misreporting costs is relatively low: the media outlet fully reveals its private information about quality but the proposals of both candidates collapse to the outlet's preferred policy. If the quality of the elected candidate is of little importance with respect to the implemented policy, then the voter might be better off without the media outlet: in this case, both candidates would pander to the uninformed voter by offering her preferred policy, and the voter would randomly cast a ballot for one of the two candidates. The next result shows conditions under which the voter is better off without the media outlet.
\begin{corollary}\label{cor:nomedia}
	If $-\gamma(\varphi_v-\varphi_m)^2+\phi/4<0$, then there exists a finite $k'>\bar k/4$ such that the voter is strictly better off without media outlet for all $k\in\left(0,k'\right)$.
\end{corollary}
Alternatively, the voter might be better off with a ``pre-election silence'' period that forbids the delivery of policy-relevant news in the run-up to the election.\footnote{Several countries operate an ``election silence'' period where no campaigning, polling, or endorsement of candidates is allowed in the period preceding a general or presidential election.}
Conditional on the intensity of misreporting costs being low enough, the voter is better off without the media outlet if: (i) $\gamma$ is high enough, so that policies are much more important than quality; (ii) the preferred policy of the voter and the media outlet are different enough; i.e., there is a large ideological difference $|\varphi_v-\varphi_m|$; (iii) $\phi$ is small enough; that is, quality has little impact on which candidate is best. By contrast, if the costs' intensity is high enough, then the presence of the media outlet always benefits the voter. Corollary~\ref{cor:nomedia} is complementary to the similar findings of \cite{chakraborty2016} in a cheap talk setting, \cite{alonso2016} in a Bayesian persuasion framework, and \cite{boleslavsky2015information} for a non-strategic and exogenous media outlet.

\subsection{Endogenous Regulation}\label{sec:incumbent}
In the previous section, Proposition~\ref{prop:welfare} suggests that a regulator concerned about the voter's welfare should implement an intensity of misreporting costs that is as high as possible. However, regulation is often performed by actors that are neither fully detached from the political process nor have interests that are perfectly aligned with those of voters. In fact, ``fake news laws'' are mostly promulgated and discussed in parliaments where the incumbent government has substantial decisive and legislative power.\footnote{\cite{funke2018guide} provide a comprehensive list of measures recently taken by governments against online misinformation.} I first discuss the case where the incumbent candidate selects the intensity of misreporting costs.

Consider the following extension of the main model: ahead of the policy-making stage, the incumbent sets costs' intensity $k_i>0$, which is publicly observed and cannot be changed in the short run. Then, the game proceeds as described in Section~\ref{sec:model}. The incumbent, being purely office-seeking, selects $k_i$ to maximize her chances of electoral victory. From the previous analysis, we obtain that in equilibrium the incumbent wins the election with probability $\iota(k)=\frac{\phi - \tau_m(q^*(k))}{2\phi}$. I denote by $k_i^*$ the costs' intensity that maximizes $\iota(k)$. The next result shows that the incumbent candidate would select a costs' intensity that is relatively low.
\begin{lemma}\label{lemma:increg}
	The incumbent sets costs' intensity $k_i^*\in \left(0,\bar k/4\right]$, where $\iota\left(k_i^*\right)=\frac{1}{2}=\lim_{k\to\infty}\iota(k)$.
\end{lemma}

To see the intuition behind this result, recall that the sequential nature of the policy-making process allows the challenger to offer policies that are more appealing to the media outlet relative to those offered by the incumbent (Proposition~\ref{prop:eqmpol}). By obtaining the outlet's support, the challenger enjoys a second-mover advantage as she becomes more likely to be elected than the incumbent.\footnote{The incumbency disadvantage effect that is behind the result in Lemma~\ref{lemma:increg} is present also in equilibria that are non-sender-preferred. By definition, $\iota(k)$ is the same even when the challenger does not break indifference in favor of the media outlet.} Figure~\ref{fig:incumbent} shows that the incumbent's probability of electoral victory is less than a half for all finite $k>\bar k /4$. Lemma~\ref{lemma:increg} shows that, when in charge of regulation, the incumbent eliminates the challenger's second-mover advantage by setting relatively low misreporting costs to force policy-convergence: when candidates advance the same policy, the media outlet never engages in misreporting, and thus the challenger cannot benefit from the outlet's support. Any higher intensity of misreporting costs would generate policy divergence and thus a conflict of interest that would benefit the challenger at the expense of the incumbent's probability of electoral victory.

Lemma~\ref{lemma:increg} casts a negative light on the process of regulation. From the voter's viewpoint, the incumbent could not select a worse costs' intensity: even though $k_i^*$ is such that misreporting and persuasion never take place, the induced policy distortion is maximized and the voter's welfare is at its minimum. The resulting outcome is \emph{as if} the media outlet could directly decide upon which candidate gets elected and which policy is implemented. Moreover, for such a low intensity of misreporting costs $k_i^*$ the voter might be better off without the media outlet at all (Corollary~\ref{cor:nomedia}). The situation is better, but still far from ideal, when the challenger is in charge of selecting the costs' intensity: in this case, the challenger maximizes her chances of electoral victory by selecting $k_c^*=\argmin_{k\in\mathbb{R}_{+}}\iota(k)>\bar k$. This level of costs' intensity generates policy divergence and thus a positive persuasion rate $\chi(k)$. However, the voter is better off with $k^*_c$ than with $k^*_i$ because of a reduction in policy distortion. As long as candidates have an influence over the regulatory process, their office-seeking motivation results in a pull for implementing costs' intensities that are lower than the voter's optimal.

Lemma~\ref{lemma:increg} also shows that the incumbent's probability of electoral victory gets close to half for arbitrarily large intensities of misreporting costs. By selecting costs' intensity $k_i^*\in\left(0,\bar k/4\right]$, the incumbent deliberately compromises the voter's welfare in order to increase her chances of winning just by an arbitrarily small amount. This strategy is arguably unappealing to voters, and it would be fair to assume that such behavior might eventually backfire with a substantial drop of consensus. On the other hand, interventions that impose arbitrarily high misreporting costs might be frowned upon if their implementation costs are large or when such stringency is perceived as a potential threat to freedom of speech. To incorporate these realistic elements in the present analysis, consider the following alternative extension: the voter has a preferred costs' intensity $k_v$ that is relatively high but finite, $k_v\geq k_c^*$; the incumbent's probability of electoral victory is $\hat\iota(k_i)=\iota(k_i)+\nu(k_i)$, where $\nu(\cdot)$ indicates how the incumbent's choice of $k_i$ affects her chances of winning the election. Suppose that $\nu(\cdot)$ is maximized for $k_i=k_v$, continuously differentiable in $k_i$, and $\nu(k')>\nu(k'')$ for all $k',k''$ such that $|k_v-k'|<|k_v-k''|$. For concreteness, say that $\nu(k_i)=y+x\cdot \phi(k_i;k_v,\sigma)$, where $y\in\mathbb{R}$, $x>0$, and $\phi(k;k_v,\sigma)$ is the probability density function of a normal distribution with mean $k_v$ and standard deviation\footnote{Clearly, parameters $y$, $x$, and $\sigma$ must respect $\hat\iota(k_i)\in[0,1]$ for every $k_i>0$.} $\sigma$. Even though I do not endogenize the mechanism through which the candidates' probability of victory is affected by the process of regulation, this alternative extension can offer some additional insights. When the choice of $k_i$ does not affect much the incumbent's chances of victory (i.e., when $x$ is low and $\sigma$ is high), the incumbent selects $k_i^*\approx\bar k/4$ as in the baseline extension; otherwise, the incumbent's optimal choice is even higher\footnote{In this last case we obtain that $k_i^*>k_v$ because $\frac{\partial \nu(k_i)}{\partial k_i}\big|_{\substack{k=k_v}}=0$ and $\frac{\partial \iota(k_i)}{\partial k_i}>0$ for all $k_i>k_v\geq k_c^*$. When setting an intensity of misreporting costs that is marginally higher than $k_v$, the incumbent increases her chances of victory by inducing more similar policies and thus reducing the conflict of interest and persuasion. For a similar reason, in this case the challenger would select an intensity of misreporting costs that is still lower than $k_v$.} than the costs' intensity preferred by the voter, i.e., $k_i^*>k_v$. Figure~\ref{fig:incumbent} provides two graphical examples of this additional extension.
\begin{figure}
	\centering
	\includegraphics[width=0.7\linewidth,trim={0 5.4cm 0 0},clip]{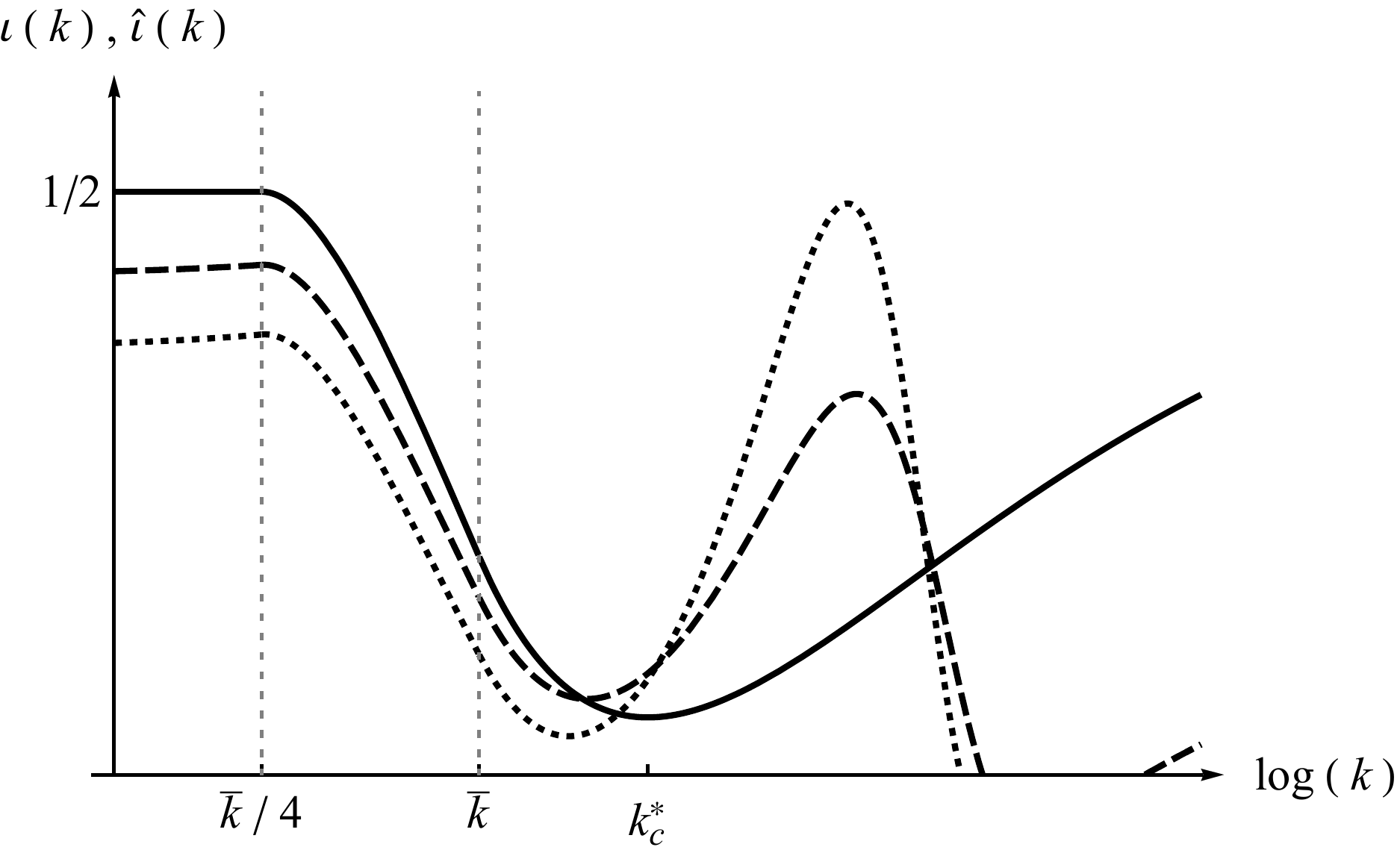}
	\caption{The incumbent's probability of electoral victory for different choices of costs' intensity. The black line represents $\iota(k)$, while the dashed and dotted lines represent $\hat\iota(k)=\iota(k)+\nu(k)$, where $\nu(k)=-.006+.2\cdot \phi(k;k_v,\sigma)$. In the dotted line, $\phi(\cdot)$ has a standard deviation of $\sigma=6$ and $\hat \iota(k)$ has a global maximum at $k^*\approx 10.5>k_v=10$; in the dashed line, $\phi(\cdot)$ has a standard deviation of $\sigma=8$ and $\hat \iota(k)$ has a global maximum at $k^*\approx \bar k/4=.25$ As the intensity of misreporting costs $k$ grows arbitrarily large, $\iota (k)$ monotonically converges to $1/2$.}
	\label{fig:incumbent}
\end{figure}

The above analysis suggests that if an electorate is highly concerned and responsive to the problem of ``fake news,'' then incumbent governments might push for extreme and disproportionate interventions; otherwise, regulation might be overly lenient. This result seems to fit with the dual reaction to recent efforts made by governments against fake news and misinformation. In some countries there is a growing consensus that governments' efforts are insufficient. With regards to the US, ``calls for regulation without censorship have been made by many people and many groups — it’s just that there is simply no political will to make a real change'' \citep{applebaum_2018}. On the other hand, there is a concurrent concern that some interventions are excessively stringent and can be exploited by governments for purely instrumental reasons. For example, the 2018 French anti-misinformation law endorsed by President Macron has received a pushback from the opposition party based on the argument that the law falls short of the principle of proportional justice. ``As regards the French solution, there seems to be a clear risk that an incumbent government constrains the freedom of expression of its opponents'' \citep{alemanno_2018}. Similarly, the German Network Enforcement Act (or NetzDG) has been criticized by Reporters without Borders and the UN Special Rapporteur on Freedom of Opinion for violating the right to freedom of the press and endangering human rights. The Act, which imposes fines of up to \euro50 million, has been recalled for revision because too much news content was blocked. ``Even the minister of justice -- who helped author the NetzDG -- had his tweets censored'' \citep{funke2018guide}.

\section{Discussion and Extensions}\label{sec:extensions}

In this section, I discuss some assumptions and extensions of the model to examine the robustness of the results. Formal proofs are relegated to the Supplementary Appendix~\ref{app:suppl}.

{\bf Media Competition.} A monopoly best describes those cases where media outlets can have exclusive possession of policy-relevant news. Interventions that alter misreporting costs may discipline these media outlets, whereas interventions that change concentration levels do not. Whilst the analysis of a monopolistic news market is an important and natural first step to study interventions that change misreporting costs, it would be interesting to also study competitive environments. In \cite{vaccari2021competition}, I extend the analysis of the communication subgame (Section~\ref{sec:commsub}) to the case of two competing media outlets that have common information and conflicting interests. The addition of a policy-making stage in a competitive environment is a challenging task because candidates can strategically create asymmetries in the conflict of interest between the media outlets and the voter. In equilibrium, the voter allocates the burden of proof across outlets accounting for their asymmetries. I leave the endeavor of endogenizing policy choices within a model of communication with multiple media outlets for future research.

{\bf Political Gains.} If $\xi$ are resources that the winning candidate inefficiently allocate to the media outlet as a reward for her endorsement, then we should take this into account when calculating the voter's welfare. Recall that, since equilibrium policies $q^*(k)$ satisfy condition~\eqref{eq:fullpers} for every finite $k$, the media outlet always exerts full persuasion and the candidate endorsed by the outlet is always elected. Therefore, we can simply subtract $\xi$ from the equilibrium welfare of the voter $W_v^*$, and all results would be unchanged.

{\bf Regulation Costs.} In this paper, I abstract from regulation costs. However, carrying out interventions typically comes at a cost that should be accounted for when computing the voter's welfare. Since results about welfare would be highly dependent from the assumptions we make on regulation costs, I leave these costs out from the current analysis.\footnote{Moreover, it is not clear which functional form would be more appropriate to describe regulation costs. For example, interventions that \emph{decrease} the costs' intensity $k$ may free up resources but they would also require time and effort.} Intuitively, assuming sufficiently high regulation costs would break down the monotonicity result in Proposition~\ref{prop:welfare}. However, as noted before, Proposition~\ref{prop:welfare} implicitly shows that lenient interventions are detrimental for the voter's welfare when misreporting costs $k$ are sufficiently low. This take holds true as long as regulation costs are positive, no matter their shape and how small they are.

{\bf Ideological Media.} In the model presented in Section~\ref{sec:model}, the media outlet obtains fixed policy-independent gains $\xi>0$ when the candidate she endorses ends up winning the election. Assuming fixed gains has two advantages: first, it simplifies the analysis and allows us to obtain closed-form solutions; second, it is an adequate approach to model a media organization that is ideologically biased (because, e.g., of the ownership or editor) but that is not affected by the candidates' policies (such as, e.g., abortion laws or gay marriage).\footnote{In these cases, an ideologically biased editor is budget constrained in that she cannot use more resources than the fixed political gains $\xi$ to misreport information. Moreover, the private gains that an ideological editor or journalist obtains from policies are likely to be small or negligible next to the outlet's political gains $\xi$.} However, it is interesting to consider a variant of the model where the media outlet's payoff is directly affected by the policies in a similar way as the voter's. In this case, the outlet's gains from inducing the election of her favorite candidate depend on the candidates' policies $q$ and the outlet's bliss $\varphi_m$. I formally study this case in Appendix~\ref{app:suppl}.

Given the incumbent's proposal $q_i$, the challenger can still garner a second-mover advantage by catering to the media outlet. Importantly, imitating the incumbent is never a best reply as long as the misreporting costs are positive and finite. By contrast, when $k\to+\infty$ both policies converge to the voter's bliss $\varphi_v$, and when $k\to 0^+$ both policies converge to $\varphi_m$. Therefore, as in the main analysis, policy divergence occurs for intermediate levels of misreporting costs, and convergence occurs otherwise. This implies that the qualitative results concerning policy distortion and the rate of persuasion carry through even when the outlet's gains are policy-dependent. Levels of misreporting costs that are close to zero lead to large policy distortions but a small conflict between the outlet and the media, and therefore less misreporting and persuasion. Higher misreporting costs decrease the policy distortion but generate a larger conflict of interest, which yields more misreporting and persuasion. Since equilibrium policies are now defined implicitly, it is now more difficult to analyze the voter's welfare, which I leave for future research.

{\bf Simultaneous Policy-making.} I assume that candidates propose policies sequentially. The sequential nature of the policy-making process is important as it shows that the challenger enjoys a second-mover advantage. As a result, the incumbent can use regulation to eliminate the challenger's advantage at the expense of the voter. However, most qualitative results hold even when candidates propose their policies simultaneously. In Lemma~\ref{lemma:simultaneous}, I show that in this case policy convergence to $(\varphi_m,\varphi_m)$ occurs when $k\leq \bar k/4$, exactly as in the case where policy-making is sequential. Moreover, for intermediate values of $k$ the equilibrium of the policy-making stage must be in mixed strategies (see Lemma~\ref{lemma:nopure}). This implies that for intermediate $k$ policy divergence occurs with positive probability.\footnote{In a similar setting with simultaneous policy proposals and cheap talk communication, \cite{chakraborty2016} characterize the candidates' mixed-strategy equilibrium for different levels of conflict of interest between the media outlet and the voter.} Finally, for $k\to +\infty$ we still obtain convergence to $(\varphi_v,\varphi_v)$. Therefore, the main qualitative results about policy distortions and the rate of persuasion hold even when candidates propose their policies simultaneously rather than sequentially.

{\bf Equilibrium Selection.} Conditions that ensure equilibrium uniqueness are important to perform exercises of comparative statics. However, as it is typical in communication games, here there is not a unique equilibrium. In this paper, I mainly focus on the perfect Bayesian equilibrium that is robust to the Intuitive Criterion and preferred by the sender. There are several reasons why this is a sensible choice: first, the sender-preferred equilibrium constitutes also the voter's worse case scenario, which is key for the robust approach to policy analysis \citep{hansen2008robustness}; second, it is also the unique perfect sequential equilibrium \citep{grossman1986perfect}; third, among the perfect Bayesian equilibria that survive the Intuitive Criterion, the sender-preferred is the only one that is undefeated \citep{mailath1993belief};\footnote{See Lemmata~\ref{lemma:perfect} and \ref{lemma:undefeated} in Appendix \ref{app:suppl}.} finally, it is ``intuitive'' \citep{cho1987}. Therefore, the analysis is centered on a focal and important equilibrium that possesses a number of appealing qualities. In addition, the results about welfare would carry through even if we were to consider the whole class of perfect Bayesian equilibria that survive the Intuitive Criterion and focus on sets of equilibrium payoffs.\footnote{To see this, notice that the multiplicity of PBE in the communication subgame (Proposition~\ref{prop:monopoly} in Appendix~\ref{app:communication}) yields a convex set $\mathcal{W}(k)$ of payoffs that the voter can obtain in equilibrium (Corollary~\ref{cor:payoffset} in Appendix~\ref{app:welfare}). Since changes in the intensity of misreporting costs $k$ affect only the lower bound of $\mathcal{W}(k)$, the focus on the voter's worst-case scenario is without loss of generality.} In Appendix~\ref{app:communication}, I show all the perfect Bayesian equilibria of the communication subgame that survive the Intuitive Criterion (Proposition~\ref{prop:monopoly}).

{\bf Tie-breaking Rule.} The ``sender-preferred'' part of the equilibrium concept implies that, when indifferent, the challenger best replies to the incumbent with a policy that is closer to the outlet's bliss $\varphi_m$. From Section~\ref{sec:policymaking} we know that when $k\in(0,\bar k/4)$ the challenger strictly prefers to propose $\varphi_m$, while for finite $k\geq \bar k/4$ she is indifferent between proposing the voter's bliss $\varphi_v$ or a biased policy. Therefore, the main qualitative result would carry through even if  the challenger were to break indifference in favor to the voter rather than the media outlet: in this case, policy convergence occurs at $k\in(0,\bar k/4)$, and then  divergence occurs due to a discontinuity at $k=\bar k/4$, where also the voter's welfare and the persuasion rate would jump to higher values. As $k$ increases, the incumbent's policy becomes increasingly populist while the challenger keeps proposing $\varphi_v$. It follows that, as $k$ increases further, the voter's welfare increases while the persuasion rate decreases.

{\bf Misreporting Costs.} I assume that the media outlet incurs misreporting costs that are continuously increasing with the size of the lie. This assumption is grounded on the following observations: first, credibly reporting a gross distortion of the truth requires more time and effort than reporting just a slight stretch of the truth; second, ``bigger lies'' are more likely to be detected and also lead to a larger reputation loss. More discussion on misreporting costs is provided by \cite{kartik2009} and \cite{kartik2007}, among others.

The use of continuous misreporting costs (and thus continuous report and state spaces) is not only a sensible assumption, but it is also important for our qualitative findings. For example, consider a discrete variant of the model with three states (high, medium, low) and a fixed lying cost $k>0$ that the outlet incurs when $r\neq\theta$. In this case, small changes in $k$ would, in general, have no effect. When accounting for regulation costs, this implies that lenient interventions necessarily damage the voter's welfare. I show that this result is not a byproduct of having a discrete state space and fixed misreporting costs.

{\bf Information Structure.} The distributional assumption on $\theta$ implies that, during the policy-making stage, both the candidates and the voter have uninformative priors about the candidates' \emph{relative} quality. Information about $\theta$ is revealed to the media outlet only after the policy-making stage. There are two main interpretation of this information structure: first, it is political parties that commit to policies and then nominate their own candidates. The parties cannot acquire every piece of information about the nominees' quality, which is often hidden on purpose; second, information about the candidates' quality is publicly known, but it is not sure how their traits will advantage them in the future. In this case, $\theta$ represents the state of the world in which the winning candidate will have to operate. For more discussion about this type of information structure, see \cite{chakraborty2016}. 

\section{Conclusion}\label{sec:conclusion}
 
This article studies the voter's welfare in relation to interventions that affect media outlets' misreporting costs. The results provide a number of policy implications. As intuition would suggest, interventions that increase the costs of misreporting information never make the voter worse off. However, lenient regulatory efforts might be futile and thus wasteful when accounting for their implementation costs. In these cases, a regulator should either do nothing or enforce substantial measures. I provide conditions under which the voter is better off without a media outlet or with a period of pre-election silence.

The presence of an influential and biased media outlet generates both policy and informational distortions. As a result, higher misreporting costs might be associated with more persuasion and a worse selection of candidates, but they can still increase the voter's welfare because of a reduction in policy distortions. Therefore, regulatory efforts such as ``fake news laws'' ought not to be judged solely by their impact on misreporting behavior. This type of intervention should not be designed with the objective of reducing or eliminating misinformation: full revelation can be achieved with relatively low misreporting costs, but the induced policy distortion would minimize the voter's welfare.

Importantly, electoral incentives skew the process of regulation as politicians strategically choose interventions to maximize their own chances of electoral victory. For purely instrumental reasons, the incumbent government deliberately pursues interventions that minimize the voter's welfare. These kind of friction in the regulatory process persists even when the challenger is in charge of regulation and when the candidates' probability of electoral victory is affected in a reasonable way by which intervention they choose to advance.

This paper is a first important step toward the development of a sensible theory of regulation in news markets. I show that the process of regulation is problematic even when politicians have the option to implement an ``ideal intervention'' that maximizes the voter's welfare at no cost and without generating an additional agency problem. However, as discussed at the end of Section~\ref{sec:welfare}, there is a widespread concern that fake news laws might infringe free speech rights. It is often difficult to publicly assess and agree upon what is the underlying ``truth'' behind news reports, and thus governments can use harsh interventions to capture the media.\footnote{At the end of Section~\ref{sec:welfare}, I incorporate the idea that voters have a distaste for harsh interventions, e.g., because of excessive implementation costs. Alternatively, voters might be afraid that the government can exploit regulation to control information. However, I do not explicitly model the agency problem between the regulator and the voter, and assume as exogenous the process by which the candidates' probability of electoral victory is affected by their choice of regulation.} Proceeding from this paper's findings, the next step is to take into consideration the possibility that regulation can be instrumentally used by political actors to directly control information.\footnote{For example, regulators might instrumentally punish a media outlet that delivers unwelcome news on a false account of misinformation. Note that the type of intervention considered in the present paper does not necessarily require an assessment of the media outlet's private information. Measures such as awareness campaigns aimed at informing people about misinformation and targeted educational programs make it more difficult and thus more costly for news providers to credibly misreport information.} I leave this for future research.

\appendix

\section{Appendix}\label{sec:app}

\subsection{The Communication Subgame}\label{app:communication}

The communication subgame $\hat\Gamma$ starts after the policy-making stage, where both candidates make binding commitments to policy proposals. In this section, I assume that the proposed policies $q=(q_i,q_c)$ are such that $\tau_m(q)<\tau_v(q)$. Since policies are fixed, in this section I simplify the notation by using $\tau_j\equiv\tau_j(q)$, $\rho(\theta)\equiv\rho(\theta,q)$, $p(\theta|r,q)\equiv p(\theta|r)$, and $\beta(r)\equiv\beta(r,q)$. I use the term ``generic equilibrium" to denote a perfect Bayesian equilibrium of the communication subgame $\hat\Gamma$ that is robust to the Intuitive Criterion \citep{cho1987}. A ``sender-preferred equilibrium" of the communication subgame $\hat\Gamma$ is the generic equilibrium preferred by the media outlet, as defined in Section~\ref{sec:model}.

Proposition~\ref{prop:monopoly} builds on Lemmata~\ref{lemma:monopmonot} to \ref{lemma:monopstrat} and shows all the generic equilibria\footnote{A sufficient condition on the state space for the existence of all generic equilibria in Proposition~\ref{prop:monopoly} is, for proposals $q$ such that $\tau_v(q)>\tau_m(q)$, $\phi\geq\max\left\{ \tau_v(q)+\sqrt{\xi/k},  -\tau_m(q) \right\}$. In this section I assume that such a condition is always satisfied.} of $\hat\Gamma$. The proofs of Proposition~\ref{prop:monopoly} and of all its supporting lemmata use a general misreporting cost function $kC(r,\theta)$, where $k>0$ and $C(\cdot,\cdot)$ is continuous on $\mathbb{R}\times\Theta$ with $C(r,\theta)\geq 0$ for all $r\in\mathbb{R}$ and $\theta\in\Theta$, $C(x,x)=0$ for all $x\in\Theta$. The cost function $C(\cdot)$ satisfies $C(r,\theta)>C(r',\theta)$ if $|r-\theta|>|r'-\theta|$ for all $\theta\in\Theta$, and $C(r,\theta)>C(r,\theta')$ if $|r-\theta|>|r-\theta'|$ for all $r\in\mathbb{R}$. I redefine the functions $l(r)$ and $h(r)$ for a general cost $C(r,\theta)$ as follows: for a $r>\tau_m$, $l(r)=\max\left\{\tau_m,\min\left\{ \theta | k C(r,\theta)=\xi \right\}\right\}$; for a $r<\tau_m$, $h(r)=\min\left\{\tau_m,\max\left\{ \theta | k C(r,\theta)=\xi \right\}\right\}$. For consistency with the rest of the paper, Proposition~\ref{prop:monopoly} and Lemma~\ref{lemma:senderpm} are expressed in terms of the cost function $C(r,\theta)=(r-\theta)^2$.

The set of all the voter's pure strategy best responses to a report $r$ and posterior beliefs $p(\cdot|r)$ such that $\int_{\theta\in T} p(\theta|r)d\theta=1$ is defined as\footnote{For $T=\varnothing$, I set $B(\varnothing,r)=B(\Theta,r)$.}
\begin{displaymath}
B(T,r)=\bigcup_{p:\int_{T} p(\theta|r)d\theta=1} \argmax_{b\in\{i,c\}} \int_{\theta \in \Theta}p(\theta|r)u_v(b,\theta,{q})d\theta.
\end{displaymath} 
Fix an equilibrium outcome and let $u^*_m(\theta)$ denote the outlet's expected equilibrium payoff in state $\theta$. The set of states for which delivering report $r$ is not equilibrium-dominated for the outlet is
\begin{displaymath}
J(r)=\left\{\theta \in \Theta \Big| u^*_m(\theta) \leq \max_{b \in B(\Theta,r)}u_m(r,b,\theta,{q})\right\}.
\end{displaymath}
An equilibrium does not survive the Intuitive Criterion refinement if there exists a state $\theta'\in\Theta$ such that, for some report $r'$, $u_m^*(\theta')<\min_{b \in B(J(r'),r')}u_m(r',b,\theta',{q})$.

In Lemma~\ref{lemma:monopstrat}, I use the following notation to denote the limits of the reporting rule $\rho(\cdot)$ as $\theta$ approaches state $t$ from, respectively, above and below: $\rho^+(t)=\lim_{\theta\to t^+}\rho(\theta)$ and $\rho^-(t)=\lim_{\theta\to t^-}\rho(\theta)$.

\begin{Alemma}\label{lemma:monopmonot}
	In a generic equilibrium of $\hat\Gamma$, $\rho(\theta)$ is non-decreasing in $\theta<\tau_m$ and $\theta>\tau_m$.
\end{Alemma}
\begin{proof}
	Consider a generic equilibrium and suppose that there are two states $\theta''>\theta'>\tau_m$ such that $\rho(\theta')>\rho(\theta'')$. We can rule out that $\beta(\rho(\theta'))=\beta(\rho(\theta''))=c$, as in such case the equilibrium would prescribe $\rho(\theta')=\theta'<\theta''=\rho(\theta'')$. If $\beta(\rho(\theta'))=\beta(\rho(\theta''))=i$, then in at least one of the two states $\theta',\theta''$ the outlet could profitably deviate by delivering the report prescribed in the other state. Consider the case where $\beta(\rho(\theta'))=i$ ($c$) and $\beta(\rho(\theta''))=c$ ($i$). In equilibrium, it has to be that $\rho(\theta'')=\theta''$ ($\rho(\theta')=\theta'$). Given $\rho(\theta')>\rho(\theta'')=\theta''>\theta'$ ($\theta''>\theta'=\rho(\theta')>\rho(\theta'')$) and $C(\rho(\theta'),\theta'')<C(\rho(\theta'),\theta')$ ($C(\rho(\theta''),\theta'')>C(\rho(\theta''),\theta')$), the outlet could profitably deviate in state $\theta''$ ($\theta'$) by reporting $\rho(\theta')$ ($\rho(\theta'')$). A similar argument applies for any two states $\theta'<\theta''<\tau_m$, completing the proof.
\end{proof}

\begin{Alemma}\label{lemma:eqmcont}
	In a generic equilibrium of $\hat\Gamma$, if $\rho(\theta)$ is strictly monotonic and continuous in an open interval, then $\rho(\theta)=\theta$ for all $\theta$ in such an interval.
\end{Alemma}
\begin{proof}
	Consider a generic equilibrium and suppose that the reporting rule $\rho(\cdot)$ is strictly increasing (decreasing) and continuous in an open interval $(a,b)$, but $\rho(\theta)>\theta$ for some $\theta\in(a,b)$. There always exist an $\epsilon>0$ such that the media outlet prefers the same alternative in both states $\theta$ and $\theta-\epsilon$, and $\theta<\rho(\theta-\epsilon)<\rho(\theta)$ (resp. $\rho(\theta-\epsilon)>\rho(\theta)>\theta$). The media outlet never pays misreporting costs to implement its least preferred alternative; therefore, it must be that $\beta(\rho(\theta))=\beta(\rho(\theta-\epsilon))$. Since $C(\rho(\theta-\epsilon),\theta)<C(\rho(\theta),\theta)$ (resp. $C(\rho(\theta),\theta-\epsilon)<C(\rho(\theta-\epsilon),\theta-\epsilon)$), the media outlet has a profitable deviation in state $\theta$ (resp. $\theta-\epsilon$), contradicting that $\rho(\cdot)$ is in equilibrium.
\end{proof}

\begin{Alemma}\label{lemma:monoptruth}
	In a generic equilibrium of $\hat\Gamma$, $\rho(\theta)=\theta$ for almost every $\theta\leq\tau_m$.
\end{Alemma}
\begin{proof}
	Consider a generic equilibrium and suppose that $\rho(\theta)\neq \theta$ for all $\theta\in S$, where $S$ is an open set such that $\sup S\leq \tau_m$ and $S\subset \Theta$. Beliefs must be such that $\beta(r)=i$ for all $r\in S$. Suppose that a report $r' \in S$ is off-path. It must be that $u^*_m(\theta)\geq u_m(r',i,\theta,q)$ for all $\theta \geq \tau_m$. Since $\sup J(r')\leq \tau_m< \tau_v$ and $B(J(r'),r')=c$, the outlet can profitably deviate by reporting truthfully when $\theta=r'\in S$. Hence, all reports $r\in S$ must be on-path. To have $\beta(r')=i$ for a $r'\in S$, it must be that $\rho(\theta')=r'$ for some $\theta'\geq\tau_v$. In all states $\theta> \tau_m$ such that $\rho(\theta)\in S$, the outlet must deliver the same least expensive report $r'\in S$ such that $\beta(r')=i$. Thus, $S$ has measure zero and $\rho(\theta)=\theta$ for almost every $\theta\leq \tau_m$.
\end{proof}

\begin{Alemma}\label{lemma:disc}
	In a generic equilibrium of $\hat\Gamma$, $\rho(\cdot)$ is discontinuous at some $\theta\in\Theta$.
\end{Alemma}
\begin{proof}
	Suppose by way of contradiction that there is a generic equilibrium where $\rho(\theta)$ is continuous in $\Theta$. From Lemma~\ref{lemma:monoptruth}, we know that $\rho(\theta)=\theta$ for $\theta\leq\tau_m$. If $\rho(\theta)=\theta$ also for all $\theta > \tau_m$, then the equilibrium would be fully revealing. In such case, the outlet could profitably deviate by reporting $\tau_v$ when the state is $\theta\in(\tau_v-\epsilon,\tau_v)$ for some $\epsilon>0$. Therefore, it must be that $\rho(\theta')\neq \theta'$ for some state $\theta'>\tau_m$. By Lemma~\ref{lemma:eqmcont}, it has to be that  $\rho(\theta')<\theta'$, or otherwise $\rho(\cdot)$ would be discontinuous; therefore Lemmata~\ref{lemma:monopmonot} and \ref{lemma:eqmcont} imply that $\rho(\theta)=\rho(\theta')$ for all $\theta \in (\max\{\rho(\theta'),\tau_m\},\sup\Theta)$. There always exists a report $r'\geq \theta'$ such that $\inf J(r')\geq \max\{\rho(\theta'),\tau_m\}$. Since $\beta(\rho(\theta'))=i$, it must be that $B(J(r'),r')=i$. Therefore, there are states where the media outlet would have a profitable deviation, contradicting that a continuous $\rho(\cdot)$ can be part of a generic equilibrium.
\end{proof}

\begin{Alemma}\label{lemma:monopstrat}
	In a generic equilibrium of $\hat\Gamma$, $\rho(\cdot)$ has a unique discontinuity in state $\theta_\delta$, where $\theta_\delta \in [\tau_m,\tau_v]$. The reporting rule\footnote{Recall that $\rho^+(t)=\lim_{\theta\to t^+}\rho(\theta)$ and $\rho^-(t)=\lim_{\theta\to t^-}\rho(\theta)$.} is such that $\rho(\theta)=\rho^+(\theta_\delta)>\theta_\delta =l(\rho^+(\theta_\delta))$ for $\theta \in (\theta_\delta,\rho^+(\theta_\delta))$ and $\rho(\theta)=\theta$ for all $\theta\in(\inf \Theta,\theta_\delta)\cup[\rho^+(\theta_\delta),\sup\Theta)$.
\end{Alemma}

\begin{proof}
	I denote by $\theta_\delta$ the lowest state in which a discontinuity of $\rho(\cdot)$ occurs. By Lemmata~\ref{lemma:monoptruth} and \ref{lemma:disc}, we know that in equilibrium such a discontinuity exists and $\theta_\delta \geq \tau_m$. 
	
	Suppose that $\rho^-(\theta_\delta)\neq \theta_\delta$. If $\rho^-(\theta_\delta)< \theta_\delta$, then by Lemmata~\ref{lemma:monopmonot} and \ref{lemma:eqmcont} we have that $\rho(\theta)=\rho^-(\theta_\delta)$ for all $\theta \in (\max\{\rho^-(\theta_\delta),\tau_m\},\theta_\delta)$ and $\rho(\theta)=\theta$ for $\theta \leq \max\{\rho^-(\theta_\delta),\tau_m\}$. In equilibrium, it has to be that $\beta(\rho^-(\theta_\delta))=i$ and $\beta(r')=c$ for every off-path $r'\in (\max\{\rho^-(\theta_\delta),\tau_m\},\theta_\delta)$. Hence, every report $r'\in (\max\{\rho^-(\theta_\delta),\tau_m\},\theta_\delta)$ is equilibrium dominated for all $\theta<\theta'$, where $\theta'=\{\theta \in \Theta \,| \, C(\rho^-(\theta_\delta),\theta)=C(r',\theta)\}$. Therefore, $B(J(r'),r')=i$, and the media outlet could profitably deviate by reporting $r'$ instead of $\rho^-(\theta_\delta)$ when $\theta\in (\theta',\theta_\delta)$. Suppose now that $\rho^-(\theta_\delta)>\theta_\delta$. By Lemma~\ref{lemma:monopmonot} we have $\rho^-(\tau_m)=\tau_m$, and thus it has to be that $\theta_\delta>\tau_m$. Similarly to the previous case, in equilibrium it must be that $\rho(\theta)=\rho^-(\theta_\delta)$ for all $\theta \in (\tau_m,\theta_\delta)$. This is in contradiction to $\theta_\delta$ being the lowest discontinuity, as we would have $\rho^+(\tau_m)>\tau_m$. Therefore, in every generic equilibrium, $\rho^-(\theta_\delta)=\theta_\delta \geq \tau_m$ and $\rho(\theta)=\theta$ for $\theta <\theta_\delta$.
	
	From Lemmata~\ref{lemma:monopmonot} and \ref{lemma:eqmcont}, it follows that $\rho^+(\theta_\delta)>\theta_\delta$ and $\rho(\theta)=\rho^+(\theta_\delta)$ for every $\theta \in (\theta_\delta,\rho^+(\theta_\delta)]$: since it must be that $\beta(\rho^+(\theta_\delta))=i$, the outlet would profitably deviate by reporting $\rho^+(\theta_\delta)$ in every state $\theta \in (\theta_\delta,\rho^+(\theta_\delta)]$ such that $\rho(\theta)>\rho^+(\theta_\delta)$. To prevent other profitable deviations, $\rho^+(\theta_\delta)$ must be such that $\xi\leq kC(\rho^+(\theta_\delta),\theta)$ for $\theta\in(\tau_m,\theta_\delta)$ and $\xi\geq kC(\rho^+(\theta_\delta),\theta)$ for all $\theta\in[\theta_\delta,\rho^+(\theta_\delta)]$. Together, these conditions imply that $\theta_\delta=l(\rho^+(\theta_\delta))$. Any off-path report $r'>\rho^+(\theta_\delta)$ would be equilibrium-dominated by all $\theta\leq \rho^+(\theta_\delta)$, yielding $B(J(r'),r')=i$. Therefore, it must be that $\rho(\theta)=\theta$ for all $\theta \geq \rho^+(\theta_\delta)$, and $\rho(\theta)=\rho^+(\theta_\delta)$ for $\theta \in (\theta_\delta,\rho^+(\theta_\delta))$. 
	
	Suppose now that $\theta_\delta>\tau_v$. Given the reporting rule, posterior beliefs $p$ must be degenerate on $\theta=r$ for all $r \in [\tau_v,\theta_\delta)$. In this case, there always exists an $\epsilon>0$ such that the outlet can profitably deviate by reporting $\tau_v$ instead of $\theta$ in states $\theta\in (\tau_v-\epsilon,\tau_v)$. Therefore, $\theta_\delta \in [\tau_m,\tau_v]$.
\end{proof}

\begin{Aproposition}\label{prop:monopoly}
	A pair $(\rho(\theta), p(\theta\,|\,r))$ is a generic equilibrium of $\hat\Gamma$ if and only if, for a given $\lambda\in\left[\tau_v,\tau_v+\tfrac{1}{2}\sqrt{\frac{\xi}{k}}\right]$,
	\begin{enumerate}
		\item[i)] The reporting rule $\rho(\theta)$ is, for a $\lambda\in\left[\tau_v,\tau_v+\tfrac{1}{2}\sqrt{\frac{\xi}{k}}\right)$,
		\[
		\rho(\theta) =
		\begin{cases}
		\hat r(\lambda)=\min\left\{\lambda+\tfrac{1}{2}\sqrt{\frac{\xi}{k}},2\lambda-\tau_m\right\} & \quad \text{if } \; \theta \in \left(l\left(\hat r(\lambda)\right),\hat r(\lambda)\right)\\ 
		\theta & \quad \text{otherwise. }
		\end{cases}
		\]
		When $\lambda=\tau_v+\tfrac{1}{2}\sqrt{\frac{\xi}{k}}$, $\rho(\theta)=\hat r(\lambda)$ for $\theta\in[l(\hat r(\lambda)),\hat r(\lambda))$, and $\rho(\theta)=\theta$ otherwise.\footnote{Up to changes of measure zero in $\rho(\theta)$ due to the media outlet being indifferent between reporting $l(\hat r(\lambda))$ and $\hat r(\lambda)$ when the state is $\theta=l(\hat r(\lambda))>\tau_m$.}
		\item[ii)] Posterior beliefs $p(\theta\,|\,r)$ are according to Bayes' rule whenever possible and such that $\mathbb{E}_p [\theta \,|\hat r(\lambda)]=\lambda$, $\mathbb{E}_p [\theta \,|r]<\tau_v$ for every off-path $r$, and $p(\theta\,|\,r)$ are degenerate on $\theta=r$ otherwise.
	\end{enumerate}
\end{Aproposition}

\begin{proof}
	Given the reporting rule $\rho(\cdot)$ described in Lemma~\ref{lemma:monopstrat}, beliefs $p$ must be such that $\beta(\rho^+(\theta_\delta))=i$, and thus $\mathbb{E}_p [\theta \,|\,\rho^+(\theta_\delta)]=\frac{\rho^+(\theta_\delta)+\theta_\delta}{2}\geq \tau_v$. With square loss misreporting costs $C(r,\theta)=(r-\theta)^2$, we have that $l(\rho^+(\theta_\delta))=\max\left\{\rho^+(\theta_\delta)-\sqrt{\frac{\xi}{k}},\tau_m\right\}\leq\tau_v$. Since $\theta_\delta=l(\rho^+(\theta_\delta))\leq \tau_v$, we also obtain that $\mathbb{E}_p [\theta \,|\,\rho^+(\theta_\delta)]\leq \tau_v+\frac{1}{2}\sqrt{\frac{\xi}{k}}$. Therefore, the expectation $\mathbb{E}_p [\theta \,|\,\rho^+(\theta_\delta)]$ induced by the report $\rho^+(\theta_\delta)$ has to be between $\tau_v$ and $\tau_v+\frac{1}{2}\sqrt{\frac{\xi}{k}}$. Similarly, for a general misreporting cost function $C(r,\theta)$, the expectation $\mathbb{E}_p [\theta \,|\,\rho^+(\theta_\delta)]$ has to be between $\tau_v$ and $\frac{\tau_v+\bar r(\tau_v)}{2}$, where $\bar r(\theta)$ is defined for a $\theta>\tau_m$ as $\bar r(\theta)=\max\left\{r\in\mathbb{R} | kC(r,\theta)=\xi\right\}$. I define the pooling report $\hat r(\lambda)$ as
	\[
	\hat r(\lambda):=\left\{r\in\mathbb{R}\; | \; \mathbb{E}_f [\theta \,|\,l(r)<\theta<r]=\lambda \right\}.
	\]
	For a $\lambda\in\left[\tau_v,\frac{\tau_v+\bar r(\tau_v)}{2}\right)$, we can rewrite the reporting rule described in Lemma~\ref{lemma:monopstrat} as
	\begin{equation}\label{eq:monop1}
	\rho(\theta) =
	\begin{cases}
	\hat r(\lambda) & \quad \text{if } \; \theta \in \left(l\left(\hat r(\lambda)\right),\hat r(\lambda)\right)\\ 
	\theta & \quad \text{otherwise. }
	\end{cases}
	\end{equation}
Alternatively, \eqref{eq:monop1} can have $\rho(l(\hat r(\lambda))=\hat r(\lambda)$ as long as $l(\hat r(\lambda))>\tau_m$. If $\lambda=\frac{\tau_v+\bar r(\tau_v)}{2}$, then it must be that \eqref{eq:monop1} has $\rho(l(\hat r(\lambda)))=\hat r(\lambda)$; otherwise the outlet would profitably deviate by reporting $\tau_v$ when the state is $\theta\in(\tau_v-\epsilon,\tau_v+\epsilon)$ for some $\epsilon>0$. Since $\theta\sim\mathcal{U}$, when $C(r,\theta)=(r-\theta)^2$ we have $\hat r(\lambda)=\lambda+\tfrac{1}{2}\sqrt{\frac{\xi}{k}}$ if $l(\hat r(\lambda))>\tau_m$ and $\hat r(\lambda)=2\lambda-\tau_m$ otherwise.
	
	By applying Bayes' rule to \eqref{eq:monop1}, we obtain that posterior beliefs $p(\theta | r)$ are such that $\mathbb{E}_p [\theta \,|\, \hat r(\lambda)]=\lambda\in\left[\tau_v,\frac{\tau_v+\bar r(\tau_v)}{2}\right]$, and are degenerate on $\theta=r$ for all $r \notin \left[l\left(\hat r\left(\lambda\right)\right),\hat r\left(\lambda\right)\right)$. For every off-path report $r' \in \left(l\left(\hat r\left(\lambda\right)\right),\hat r\left(\lambda\right)\right)$ it must be that $\mathbb{E}_p [\theta \,|\, r']<\tau_v$ to have $\beta(r')=c$. These off-path beliefs are consistent with the Intuitive Criterion since for every $r' \in \left(l\left(\hat r\left(\lambda\right)\right),\hat r\left(\lambda\right)\right)$ we have that $\inf J(r')<l(\hat r(\lambda))\leq\tau_v$, and thus $c\in B(J(r'),r')$. The proof is completed by the observation that the pair $(\rho(\theta), p(\theta | r))$ described in Proposition~\ref{prop:monopoly} is indeed a generic equilibrium of $\hat\Gamma$ for every $\lambda \in \left[\tau_v,\frac{\tau_v+\bar r(\tau_v)}{2}\right]$. 
\end{proof}

\begin{proof}[\bf Proof of Lemma~\ref{lemma:senderpm}]
For the case $\tau_m<\tau_v$, Proposition~\ref{prop:monopoly} shows that there is a continuum of generic equilibria of $\hat\Gamma$ parameterized by the expectation $\lambda=\mathbb{E}[\theta | \hat r(\lambda)]$. Given costs $C(r,\theta)=(r-\theta)^2$, $\lambda\in\left[\tau_v,\tau_v+\tfrac{1}{2}\sqrt{\frac{\xi}{k}}\right]$, and $\tau_m<\tau_v$, in a generic equilibrium there is persuasion when $\theta\in(l(\hat r(\lambda)),\tau_v)$. Therefore, $\lambda=\tau_v$ maximizes the media outlet's expected equilibrium payoff: for every $\lambda\in\left(\tau_v,\tau_v+\tfrac{1}{2}\sqrt{\frac{\xi}{k}}\right]$, if $l(\hat r(\tau_v))>\tau_m$, then $l(\hat r(\lambda))>l(\hat r(\tau_v))$; if $l(\hat r(\tau_v))=\tau_m$, then $l(\hat r(\lambda))\geq\tau_m$ and $\hat r(\lambda)>\hat r(\tau_v)$. That is, in the generic equilibrium where $\lambda=\tau_v$, the media outlet is either more likely to persuade the voter at the same expected cost, or is at least equally likely to persuade the voter at a strictly lower cost compared to generic equilibria where $\lambda>\tau_v$. The sender-preferred equilibrium reporting rule $\rho(\cdot)$ and beliefs $p$ follow from Proposition~\ref{prop:monopoly}, where the case $\tau_m>\tau_v$ is obtained in a similar way as $\tau_m<\tau_v$, and the case $\tau_m=\tau_v$ follows by setting $\tau_m \to \tau_v$ in the generic equilibrium of Proposition~\ref{prop:monopoly} where $\lambda=\tau_v$.
\end{proof}

\subsection{The Policy-making Stage}\label{app:policystage}

\subsubsection{The Challenger's Best Response}\label{app:best}
Given the equilibrium of the communication subgame $\hat\Gamma$ (see Lemma~\ref{lemma:senderpm}) and a policy proposal by the incumbent $q_i$, the expected utility of the challenger is $V_c(q)=l(r^*(q))$ if $\tau_m(q)<\tau_v(q)$, and $V_c(q)=h(r^*(q))$ if $\tau_m(q)>\tau_v(q)$. We have that $\tau_m(q)=\tau_v(q)$ only if $q_c=q_i$; in this case, the challenger ensures her electoral victory half the time by mimicking the incumbent's proposal, and $V_c(q)=0$. By contrast, $\tau_v(q)>\tau_m(q)$ when $q_c>q_i$, and $\tau_v(q)<\tau_m(q)$ otherwise. I define the ``best response to the left'' $BR_c^L(q_i)$ as the best response of the challenger to policy $q_i$ subject to the constraint that $q_c\leq q_i$, that is, $BR_c^L(q_i)=\argmax_{q_c\leq q_i} V_c(q)$. The ``best response to the right'' is similarly defined as $BR_c^R(q_i)=\argmax_{q_c\geq q_i} V_c(q)$.

\begin{step}\label{step:left}
	The challenger's ``best response to the left'' $BR_c^L(q_i)$ is,
	\[
	BR_c^L(q_{i})= 
	\begin{cases} 
	q_i & \text{if } q_{i} \leq \varphi_m\\
	\varphi_m & \text{if } q_{i} \in \left[\varphi_m,\varphi_m+\frac{\sqrt{\frac{\xi}{k}}}{4\gamma(\varphi_v-\varphi_m)}\right]\\
	\tilde q_c(q_i)=q_{i}-\frac{\sqrt{\frac{\xi}{k}}}{4\gamma(\varphi_v-\varphi_m)} & \text{if } q_{i} \in \left[\varphi_m+\frac{\sqrt{\frac{\xi}{k}}}{4\gamma(\varphi_v-\varphi_m)},\varphi_v+\frac{\sqrt{\frac{\xi}{k}}}{4\gamma(\varphi_v-\varphi_m)}\right]  \\
	\varphi_v  & \text{if } q_{i} \geq \varphi_v+\frac{\sqrt{\frac{\xi}{k}}}{4\gamma(\varphi_v-\varphi_m)}
	\end{cases}
	\] 
\end{step}
\begin{proof}
	Given $q_c<q_i$ and the equilibrium in Lemma~\ref{lemma:senderpm}, the challenger wins when $\theta<h(r^*(q))$, where $h(r^*(q))=\min\left\{r^*(q)+\sqrt{\frac{\xi}{k}},\tau_m(q)\right\}$ and $r^*(q)=\max\left\{\tau_v(q)-\tfrac{1}{2}\sqrt{\frac{\xi}{k}},2\tau_v(q)-\tau_m(q)\right\}$. When $h(r^*(q))<\tau_m(q)$, the pooling report is $r^*(q)=\tau_v(q)-\tfrac{1}{2}\sqrt{\frac{\xi}{k}}$, and thus  $\frac{\partial h(r^*(q))}{\partial q_c}=2\gamma (\varphi_v-q_c)>0$ and $\frac{\partial \tau_m(q)}{\partial q_c}=2\gamma(\varphi_m-q_c)<0$ for all $q_c \in [\varphi_m,\varphi_v]$. Thus, the expected utility of the challenger $V_c(q)=h(r^*(q))$ is maximized, subject to $q_c<q_i$, when $q_c$ is such that $h(r^*(q))=\tau_m(q)$. This last equality is satisfied when $q_c=\tilde q_c(q_i)$, where
	\begin{equation}
	\tilde q_c(q_i)=q_i-\frac{\sqrt{\frac{\xi}{k}}}{4\gamma (\varphi_v-\varphi_m)}.
	\end{equation}
	Therefore, as long as $\tilde q_c(q_i)\in[\varphi_m,\varphi_v]$, we have that $BR_c^L(q_i)=\tilde q_c(q_i)$. Since policies $q_j\notin[\varphi_m,\varphi_v]$, $j\in\{i,c\}$, are never optimal, it follows that if $\tilde q_c(q_i)<\varphi_m\leq q_i$, then $BR_c^L(q_i)=\varphi_m$; if $q_i<\varphi_m$, then $BR_c^L(q_i)=q_i$; if $\tilde q_c(q_i) > \varphi_v$, then $BR_c^L(q_i)=\varphi_v$. The proof is completed by solving for these inequalities.
\end{proof}

\begin{step}\label{step:right}
	The challenger's ``best response from the right'' $BR_c^R(q_{i})$ is,
	\[
	BR_c^R(q_{i})= 
	\begin{cases} 
	\varphi_v & \text{ if } q_{i} < \varphi_v-\sqrt{\frac{1}{2\gamma}\sqrt{\frac{\xi}{k}}}\\
	q_{i} & \text{ otherwise. } 
	\end{cases}
	\]
\end{step}
\begin{proof}
	Given $q_c>q_i$ and the equilibrium in Lemma~\ref{lemma:senderpm}, the challenger wins the election when $\theta<l(r^*(q))$. Since $\frac{\partial l(r^*(q))}{\partial q_c}=\frac{\partial h(r^*(q))}{\partial q_c}$, we can proceed as in Step~\ref{step:left}: the policy $q_c$ such that $l(r^*(q))=\tau_m(q)$ minimizes $V_c(q)$ subject to $q_c\geq q_i$, and thus $BR_c^R(q_i)=\varphi_v$ as long as $l(r^*(q_i,\varphi_v))> 0$. Otherwise, the challenger would be better off by imitating the incumbent with $q_c=q_i$, thereby ensuring herself a payoff of $V_c(q)=0$. The condition $l(r^*(q_i,\varphi_v))> 0$ is satisfied by $q_i < \varphi_v- \sqrt{\frac{1}{2\gamma}\sqrt{\xi/k}}$, completing the proof.
\end{proof}

\begin{Aproposition}\label{prop:bestresp}
	The challenger's best response $BR_c(q_i)$ to a policy $q_i\in[\varphi_m,\varphi_v]$ is,
\[
BR_c(q_{i})=
\begin{cases} 
\varphi_v  & \text{if } q_{i}\in\left[\varphi_m,\varphi_v + \eta(k) - \sqrt[4]{\frac{\xi}{\gamma^2 k}}\right] \text{ and } k\geq\bar k\\
q_i-\eta(k) & \text{if } q_{i}\in\left[\varphi_v + \eta(k) - \sqrt[4]{\frac{\xi}{\gamma^2 k}},\varphi_v\right] \text{ and } k\geq\bar k  \\
\varphi_v  & \text{if } q_{i}\in\left[\varphi_m,\frac{\varphi_v+\varphi_m}{2}-\eta(k)\right] \text{ and } k\in\left(0,\bar k\right]\\
\varphi_m  & \text{if } q_{i} \in \left[\frac{\varphi_v+\varphi_m}{2}-\eta(k), \varphi_m+\eta(k)\right] \text{ and } k\in\left(0,\bar k\right]\\
q_i-\eta(k) &  \text{if } q_{i} \in\left[\varphi_m+\eta(k),\varphi_v\right]\text{ and } k\in\left(0,\bar k\right],\\
\end{cases}
\]
where $\eta (k)=\frac{\sqrt{\frac{\xi}{k}} }{4\gamma(\varphi_v-\varphi_m)}$ and $\bar k = \frac{\xi}{\gamma^2 (\varphi_v-\varphi_m)^4}$.
\end{Aproposition}
\begin{proof}
	Given a policy $q_i\in[\varphi_m,\varphi_v]$ and best responses $BR_c^L(q_i)$, $BR_c^R(q_i)$ as in Steps~\ref{step:left} and \ref{step:right}, we have that  $\frac{\partial V_c\left(q_i,BR_c^R(q_i)\right)}{\partial q_i}\leq 0 \leq \frac{\partial V_c\left(q_i,BR_c^L(q_i)\right)}{\partial q_i}$. Therefore, if there is a $q_i'\in[\varphi_m,\varphi_v]$ such that $V_c\left(q_i',BR_c^R(q_i')\right)=V_c\left(q_i',BR_c^L(q_i')\right)$, then $BR_c(q_i)=BR_c^R(q_i)$ for all $q_i\in[\varphi_m,q_i']$ and $BR_c(q_i)=BR_c^L(q_i)$ for all $q_i\in[q_i',\varphi_v]$. As a first step, I compare $V_c\left(q_i,\tilde q_c(q_i)\right)$ and $V_c\left(q_i, \varphi_v\right)$. When $q_c<q_i$ and $(q_c-q_i)^2\leq\frac{\xi}{16\gamma^2 k (\varphi_m-\varphi_v)^2}$, we have that $h(r^*(q_i,q_c))=\tau_m(q_i,q_c)$, and therefore $V_c\left(q_i,\tilde q_c(q_i)\right)=\tau_m(q_i,\tilde q_c(q_i))$ and $V_c\left(q_i, \varphi_v\right)=l(r^*(q_i,\varphi_v))=\gamma(\varphi_v-q_i)^2-\tfrac{1}{2}\sqrt{\tfrac{\xi}{k}}$. The challenger's expected utility from ``best replying to the left'' with $\tilde q_c(q_i)$ is
	\begin{displaymath}
	V_c\left(q_i,\tilde q_c(q_i)\right)= \frac{1}{2}\sqrt{\frac{\xi}{k}} - \gamma \left[2 (\varphi_v-q_i)+\frac{\sqrt{\frac{\xi}{k}} }{4\gamma(\varphi_v-\varphi_m)}\right]\frac{\sqrt{\frac{\xi}{k}} }{4\gamma(\varphi_v-\varphi_m)}.
	\end{displaymath}
	Thus, the condition $\tau_m(q_i,\tilde q_c(q_i))=l(r^*(q_i,\varphi_v))$ can be rewritten as
	\begin{displaymath}
	\gamma(\varphi_v-q_i)^2 + 2\gamma \frac{\sqrt{\frac{\xi}{k}} }{4\gamma(\varphi_v-\varphi_m)}(\varphi_v-q_i)+\gamma\left(\frac{\sqrt{\frac{\xi}{k}} }{4\gamma(\varphi_v-\varphi_m)}\right)^2-\sqrt{\tfrac{\xi}{k}} =0.
	\end{displaymath}
	By solving a quadratic equation in $(\varphi_v-q_i)$, I obtain that the threshold $\bar q'$ such that $V_c\left(\bar q',BR_c^R(\bar q')\right)=V_c\left(\bar q',BR_c^L(\bar q')\right)$ is,
	\begin{displaymath}
	\bar q' = \varphi_v + \frac{\sqrt{\frac{\xi}{k}}}{4\gamma(\varphi_v-\varphi_m)} - \sqrt[4]{\frac{\xi}{\gamma^2 k}}.
	\end{displaymath}
	Since $V_c(q_i,q_i)=0$, I do not need to consider the case where $BR_c^R(q_i)=q_i$ as the challenger can always get a positive expected utility $V_c(q_i,q_c')=\gamma(q_c'-q_i)^2\geq 0$ by proposing $q_c'=\max\{\varphi_m,\tilde q_c(q_i)\}$. Since $BR_c^L(q_i)=\varphi_m$ when $\tilde q_c(q_i)<\varphi_m$ and $q_i\in[\varphi_m,\varphi_v]$, the comparison between $V_c(q_i,\tilde q_c(q_i))$ and $V_c(q_i,\varphi_v)$ makes sense as long as $\tilde q_c(q_i)\geq \varphi_m$ for all $q_i\in[\bar q',\varphi_v]$.  Given that $\frac{\partial \tilde q_c(q_i)}{\partial q_i}=1$, the condition is $\tilde q_c(\bar q')\geq \varphi_m$ or $k\geq \bar k$, where
	\[
	\bar k= \frac{\xi}{\gamma^2 (\varphi_v-\varphi_m)^4}.
	\]
	If $k\in\left(0,\bar k\right)$, then we have that $\tilde q_c(q_i)<\varphi_m$ and thus $BR_c^L(q_i)=\varphi_m$ for some $q_i\geq \bar q'$. In this case, the relevant comparison is between $V_c(q_i,\varphi_m)=
	\tau_m(q_i,\varphi_m)$ and $V_c(q_i,\varphi_v)$: by equating $\tau_m(q_i,\varphi_m)=l(r^*(q_i,\varphi_v))$ we get that the threshold is
	\begin{displaymath}
	\bar q''= \frac{\varphi_v+\varphi_m}{2}-\frac{\sqrt{\frac{\xi}{k}}}{4\gamma(\varphi_v-\varphi_m)}.
	\end{displaymath}
	Note that $\bar q ' = \bar q ''=\varphi_m+\frac{\sqrt{\frac{\xi}{k}}}{4\gamma(\varphi_v-\varphi_m)}$ when $k=\bar k$. Therefore, when $k\in\left(0,\bar k\right)$ we have that $BR_c(q_i)=\varphi_v$ for all $q_i\in[\varphi_m,\bar q'']$ and $BR_c(q_i)=BR_c^L(q_i)$ for $q_i\in\left[\bar q'',\varphi_v\right]$. Moreover, $BR_c^L(q_i)=\tilde q_c(q_i)$ as long as $\tilde q_c(q_i)\geq \varphi_m$, and  $BR_c^L(q_i)=\varphi_m$ otherwise. We have that $\tilde q_c(q_i)\geq \varphi_m$ when $q_i\geq \bar q'''$, where   
	\[
	\bar q''' = \varphi_m+\frac{\sqrt{\frac{\xi}{k}}}{4\gamma(\varphi_v-\varphi_m)}.
	\]	 
	The proposition follows by replacing $\eta (k)=\frac{\sqrt{\frac{\xi}{k}} }{4\gamma(\varphi_v-\varphi_m)}$.
\end{proof}

\subsubsection{Equilibrium Policy-making}\label{app:policy}

\begin{proof}[\bf Proof of Proposition~\ref{prop:eqmpol}]
	I denote by $\hat V_i(q_i)\equiv V_i(q_i,BR_c(q_i))$ the utility of the incumbent given that $q_i\in[\varphi_m,\varphi_v]$ and $q_c=BR_c(q_i)$, where the challenger's best response $BR_c(q_i)$ is as in Proposition~\ref{prop:bestresp}. Since an equilibrium is a sender-preferred PBE, when the challenger is indifferent between some policies, she selects the policy that is closer to the media outlet's bliss $\varphi_m$. Given that $h(r^*(q_i,q_c))=\tau_m(q_i,q_c)$ for $q_c=\max\{\varphi_m,\tilde q_c(q_i)\}$, we have that
	\[
	\hat V_i(q_i)=
	\begin{cases}
	- l(r^*(q_i,\varphi_v)) &  \text{if } q_{i}\in\left[\varphi_m,\varphi_v + \eta(k) - \sqrt[4]{\frac{\xi}{\gamma^2 k}}\right) \text{ and } k\geq\bar k\\
	
	-\tau_m(q_i,\tilde q_c(q_i)) & \text{if } q_{i}\in\left[\varphi_v + \eta(k) - \sqrt[4]{\frac{\xi}{\gamma^2 k}},\varphi_v\right] \text{ and } k\geq\bar k  \\
	-l(r^*(q_i,\varphi_v)) & \text{if } q_{i}\in\left[\varphi_m,\frac{\varphi_v+\varphi_m}{2}-\eta(k)\right) \text{ and } k\in\left(0,\bar k\right]\\
	
	-\tau_m(q_i,\varphi_m) &  \text{if } q_{i} \in \left[\frac{\varphi_v+\varphi_m}{2}-\eta(k), \varphi_m+\eta(k)\right] \text{ and } k\in\left(0,\bar k\right]\\
	
	- \tau_m(q_i,\tilde q_c(q_i)) & \text{if } q_{i} \in\left[\varphi_m+\eta(k),\varphi_v\right]\text{ and } k\in\left(0,\bar k\right],\\
	\end{cases}
	\]
	where $\eta (k)=\frac{\sqrt{\xi/k} }{4\gamma(\varphi_v-\varphi_m)}$ and $\bar k = \frac{\xi}{\gamma^2 (\varphi_v-\varphi_m)^4}$. Henceforth, I will use the following notation: $\bar q'=\varphi_v + \eta(k) - \sqrt[4]{\frac{\xi}{\gamma^2 k}}$, $\bar q''=\frac{\varphi_v+\varphi_m}{2}-\eta(k)$, and $\bar q'''=\varphi_m+\eta(k)$.
	
	When $k\geq\bar k$, the utility $\hat V_i(q_i)$ is increasing in $q_i$ until $q_i=\bar q'$, and decreasing afterwards, as $\frac{\partial \hat V_i(q_i)}{\partial q_i}=2\gamma(\varphi_v-q_i)>0$ for $q_i\in[\varphi_m,\bar q']$ and $\frac{\partial \hat V_i(q_i)}{\partial q_i}=2\gamma(\varphi_m-q_i)<0$ for $q_i\in[\bar q',\varphi_v]$. Since $- l(r^*(\bar q',\varphi_v))=-\tau_m(q_i,\tilde q_c(q_i))$, it follows that $q_i=\bar q'$ maximizes $\hat V_i(q_i)$ for $k\geq \bar k$. The challenger replies to $q_i=\bar q'$ with the sender-preferred policy $q_c=\tilde q_c(\bar q')$.
	
	There are three different configurations to consider when the misreporting costs are lower than $\bar k$: (i) when $\tfrac{\bar k}{4}\leq k<\bar k$, the relevant thresholds are contained within the bliss policies of the voter and the media outlet, $\varphi_m\leq\bar q''<\bar q'''<\varphi_v$; (ii) when $\tfrac{\bar k}{16}\leq k<\tfrac{\bar k}{4}$, the threshold $\bar q''$ is lower than the media outlet's bliss $\varphi_m$, and we have $\bar q''<\varphi_m<\bar q'''\leq \varphi_v$; (iii) when $0<k<\tfrac{\bar k}{16}$, both thresholds are beyond the bliss policies,  $\bar q''<\varphi_m<\varphi_v<\bar q'''$.

	In the first case, where $k \in \left[\bar k /4,\bar k\right)$, we have that $\frac{\partial \hat V_i(q_i)}{\partial q_i}=\frac{\partial -l(r^*(q_i,\varphi_v))}{\partial q_i}=2\gamma (\varphi_v-q_i)>0$ for $q_i\in[\varphi_m,\bar q'']$; $\frac{\partial \hat V_i(q_i)}{\partial q_i}=\frac{\partial -\tau_m(q_i,\varphi_m)}{\partial q_i}=2\gamma (\varphi_m-q_i)<0$ for $q_i\in[\bar q'',\bar q''']$; and $\frac{\partial \hat V_i(q_i)}{\partial q_i}=\frac{\partial -h(r^*(q_i,\tilde q_c(q_i)))}{\partial q_i}=-\frac{\sqrt{\xi/k}}{2(\varphi_v-\varphi_m)}<0$ for $q_i\in[\bar q''',\varphi_v]$. Since $-l(r^*(\bar q'',\varphi_v))=-\tau_m(\bar q'',\varphi_m)$ and $-h(r^*(q_i,\tilde q_c(q_i)))=-\tau_m(\bar q''',\varphi_m)$, we have that when $k \in \left[\bar k /4,\bar k\right)$ the incumbent maximizes $\hat V_i(q_i)$ by selecting $q_i=\bar q''$, and the challenger best responds to $\bar q''$ by proposing the sender-preferred policy $q_c=\varphi_m$.
	
	The same line of reasoning can be extended to the other two cases: when $\tfrac{\bar k}{16}\leq k<\tfrac{\bar k}{4}$, the incumbent proposes $q_i=\bar q''$ and the challenger replies with $q_c=\varphi_m$; when $0<k<\tfrac{\bar k}{16}$, both the incumbent and the challenger propose $q_j=\varphi_m$, $j\in\{i,c\}$. The proposition follows by denoting $q^*_i(k)=\argmax_{q_i\in\mathbb{R}}\hat V_i(q_i)$ and $q_c^*(q_i,k)=\min BR_c(q_i)$.
\end{proof}

\begin{Acorollary}\label{cor:support}
	The equilibrium in Proposition~\ref{prop:eqmpol} exists if and only if
	\[
	\phi\geq\min\left\{ \gamma(\varphi_v-\varphi_m)^2+\tfrac{1}{2}\sqrt{\xi/k}, 3\gamma(\varphi_v-\varphi_m)^2 \right\}.
	\] 
\end{Acorollary}
\begin{proof}
	Consider the sender-preferred equilibrium of $\hat\Gamma$ in Lemma~\ref{lemma:senderpm} and the equilibrium policies in Proposition~\ref{prop:eqmpol}. When $k\in(0, \bar k/4]$, we have that $q^*(k)=(\varphi_m,\varphi_m)$. Suppose that the challenger deviates from the prescribed equilibrium strategy by proposing $q_c=\varphi_v$. If $\phi<r^*(\varphi_m,\varphi_v)$, then there is no report that can convince the voter to cast a ballot for the incumbent, and the deviation would be profitable. Therefore, to ensure the existence of an equilibrium as in Proposition~\ref{prop:eqmpol}, it is necessary that $\phi\geq r^*(\varphi_m,\varphi_v)$. Given that for $q$ such that $q_j\in[\varphi_m,\varphi_v]$, $j\in\{i,c\}$, $\tau_v(q)$ is maximized when $q_i=\varphi_m$ and $q_c=\varphi_v$, the condition is also sufficient. The proof is completed by $\tau_v(\varphi_m,\varphi_v)=\gamma(\varphi_v-\varphi_m)^2=-\tau_m(\varphi_m,\varphi_v)$.
\end{proof}

\subsection{Voter's Welfare}\label{app:welfare}
To ease notation, in this section I will use $q_c^*(q_i^*(k),k)\equiv q_c^*(q_i^*(k))$.
\begin{proof}[\bf Proof of Proposition~\ref{prop:welfare}]
	Proposition~\ref{prop:eqmpol} shows that equilibrium policies $q^*(k)$ are such that $q^*_i(k)\geq q_c^*(q_i^*(k))$ for every $k>0$. Moreover, since $\left(q_c^*\left(q_i^*(k)\right)-q_i^*(k)\right)^2\leq\frac{\xi}{16\gamma^2 k (\varphi_m-\varphi_v)^2}$, we have that $h(r^*(q^*(k)))=\tau_m(q^*(k))$ for every $k>0$. Given the equilibrium of the communication subgame $\hat\Gamma$ (Proposition~\ref{prop:monopoly} and Lemma~\ref{lemma:senderpm}) and that $\theta\sim\mathcal{U}[-\phi,\phi]$, the incumbent wins with ex-ante probability $\frac{\phi-\tau_m(q^*(k))}{2\phi}$. When the incumbent is elected, the voter receives an expected utility of $-\gamma(\varphi_v-q_i^*(k))^2+\mathbb{E}_f\left[\theta | \theta>\tau_m(q^*(k))\right]$, where $\mathbb{E}_f\left[\theta | \theta>\tau_m(q^*(k))\right]=\frac{\phi+\tau_m(q^*(k))}{2}$. When the challenger is elected, the voter obtains a utility of $ -\gamma (\varphi_v-q_c^*(q_i^*(k)))^2$. Therefore, the voter's equilibrium welfare can be written as
	\begin{equation}\label{eq:welfare}
	\begin{split}
	W_v^*(k)& =  \left(\frac{\tau_m({ q^*(k)})+\phi}{2\phi}\right) \left[ -\gamma (\varphi_v-q_c^*(q_i^*(k)))^2 \right] \\
	& +   \left(\frac{\phi - \tau_m({q^*(k)})}{2\phi}\right) \left[ -\gamma(\varphi_v-q_i^*(k))^2 + \frac{\phi + \tau_m({ q^*(k)})}{2}\right].
	\end{split}
	\end{equation}
	
	When $k\in\left(0,\bar k/4\right)$, since $\tau_m(q)=0$ for $q=(\varphi_m,\varphi_m)$, equation \eqref{eq:welfare} reduces to $W_v^*(k)= -\gamma (\varphi_v-\varphi_m)^2+\phi /4$. Therefore, the voter's equilibrium welfare $W_v^*(k)$ is independent of $k$ for all $k\in\left(0,\bar k/4\right)$.
	
	Consider now the case where $k \in \left[\bar k /4, \bar k\right]$. The derivative of the voter's welfare with respect to the misreporting costs $k$ is
	\begin{equation}\label{eq:parwelfare1}
	\begin{split}
	\frac{\partial W_v^*(k)}{\partial k}&= \left(\frac{\phi - \tau_m({\ q^*(k)})}{2\phi}\right) \left[ \frac{1}{2}\frac{\partial \tau_m({ q^*(k)})}{\partial k}-\frac{\partial \tau_v({q^*(k)})}{\partial k} \right]\\
	&-\left( \frac{1}{2\phi}\frac{\partial \tau_m({ q^*(k)})}{\partial k} \right) \left[ \frac{\phi+\tau_m({q^*(k)})}{2} -\tau_v({ q^*(k)}) \right].
	\end{split}
	\end{equation}
	
	For $k\in \left[\bar k /4, \bar k\right]$, we obtain the following derivatives: $\frac{\partial q_i^*(k)}{\partial k}= \frac{1}{4\gamma(\varphi_v-\varphi_m)}\frac{1}{2\sqrt{\frac{\xi}{k}}}\frac{\xi}{k^2}>0$, $\frac{\partial\tau_m({ q^*(k)})}{\partial k}= 2\gamma(q_i^*(k)-\varphi_m)\frac{\partial q_i^*(k)}{\partial k}>0$, and $\frac{\partial\tau_v({ q^*(k)})}{\partial k}= -2\gamma(\varphi_v-q_i^*(k))\frac{\partial q_i^*(k)}{\partial k}<0$. Moreover, notice that $\tau_v({q^*(k)})-\tau_m({q^*(k)})=2\gamma(\varphi_m-q_i^*(k))(\varphi_v-\varphi_m)$ and $\tau_m({q^*(k)})=\gamma(\varphi_m-q_i^*(k))^2$. Therefore, equation \eqref{eq:parwelfare1} can be rewritten as
	\begin{equation}\label{eq:parwelfare2}
	\frac{\partial W_v^*(k)}{\partial k}=\frac{\gamma}{\phi}\frac{\partial q_i^*(k)}{\partial k}\left[(\phi-\tau_m({q^*(k)}))(\varphi_v-q_i^*(k)) -2\gamma(q_i^*(k)-\varphi_m)^2 (\varphi_v-\varphi_m) \right].
	\end{equation}
	
	As $k$ increases within $\left[\bar k /4,\bar k\right]$, the term $(\phi-\tau_m({q^*(k)}))(\varphi_v-q_i^*(k))$ continuously decreases while the term $(q_i^*(k)-\varphi_m)^2(\varphi_v-\varphi_m)$ continuously increases. Therefore, the derivative in equation \eqref{eq:parwelfare2} is decreasing in $k$ as $\frac{\gamma}{\phi}\frac{\partial q_i^*(k)}{\partial k}>0$ and $\frac{\partial^2 q_i^*(k)}{\partial k^2}<0$. Hence, to show that $\frac{\partial W_v^*(k)}{\partial k}>0$ for all $k\in\left[\bar k /4, \bar k\right]$, it is sufficient to show that $\frac{\partial W_v(k)}{\partial k}\big|_{\substack{k=\bar k}}>0$. Since by assumption $\phi> \gamma(\varphi_v-\varphi_m)^2$, I replace $\phi$ by $\gamma(\varphi_v-\varphi_m)^2$ and $q_i^*\left(\bar k\right)$ by $\frac{\varphi_v+3\varphi_m}{4}$ in equation \eqref{eq:parwelfare2} to obtain that
	\begin{displaymath}
	\left[\gamma\left(\varphi_v-\varphi_m\right)^2-\tau_m\left(q_i^*\left(\bar k\right),\varphi_m\right)\right]\left(\varphi_v-q_i^*\left(\bar k\right)\right) -2\gamma\left(q_i^*\left(\bar k\right)-\varphi_m\right)^2 \left(\varphi_v-\varphi_m\right)>0.
	\end{displaymath}
	Therefore, the voter's welfare $W_v(k)$ is strictly increasing in $k$ for every $k\in\left[\bar k/4,\bar k\right]$.
	
	Consider now the case where the misreporting costs are relatively high, $k\geq \bar k$. I rewrite the welfare function in equation \eqref{eq:welfare} by explicitly separating the expected gains from quality, i.e.,
	\begin{equation}
	\begin{split}
	W_v^*(k) &=  \left(\frac{\tau_m({ q^*(k)})+\phi}{2\phi}\right) \left[ -\gamma (\varphi_v-q_c^*(q_i^*(k)))^2 \right]\\
	&+   \left(1-\frac{\tau_m({ q^*(k)})+\phi}{2\phi} \right) \left[ -\gamma(\varphi_v-q_i^*(k))^2\right] + \frac{\phi^2-\tau_m^2(q^*(k))}{4\phi}.
	\end{split}
	\end{equation}
	The threshold $\tau_m(q^*(k))=\gamma\left(2\varphi_m-q_c^*\left(q_i(k)\right)-q_i^*(k)\right)\left(q_c^*\left(q_i(k)\right)-q_i^*(k)\right)$ is positive as $q_i^*(k)>q_c^*(q_i(k))\geq \varphi_m$ for every finite $k\geq \bar k$. I write the derivative $\frac{\partial \tau_m(q^*(k))}{\partial k}$ as
	\begin{displaymath}
	\begin{split}
	\frac{\partial \tau_m(q^*(k))}{\partial k} & =  \gamma\bigg[ \left(-\frac{\partial q_c^*(q_i^*(k))}{\partial k}-\frac{\partial q_i^*(k)}{\partial k}\right)\left(q_c^*(q_i^*(k))-q_i^*(k)\right)\\
	& + \left(2\varphi_m-q_c^*(q_i^*(k))-q_i^*(k)\right)\left(\frac{\partial q_c^*(q_i^*(k))}{\partial k}-\frac{\partial q_i^*(k)}{\partial k}\right)\bigg]\\
	& = \frac{\gamma}{\left(q_c^*(q_i^*(k))-q_i^*(k)\right)}\bigg[ \left(-\frac{\partial q_c^*(q_i^*(k))}{\partial k}-\frac{\partial q_i^*(k)}{\partial k}\right)\left(q_c^*(q_i^*(k))-q_i^*(k)\right)^2\\
	& + \tau_m(q^*(k))\left(\frac{\partial q_c^*(q_i^*(k))}{\partial k}-\frac{\partial q_i^*(k)}{\partial k}\right)\bigg]<0,
	\end{split}
	\end{displaymath}
	where we obtain $\frac{\partial \tau_m(q^*(k))}{\partial k}<0$ because, for every finite $k\geq \bar k$, it is the case that $q_c^*(q_i^*(k))<q_i^*(k)$, $\frac{\partial q_c^*(q_i^*(k))}{\partial k}>\frac{\partial q_i^*(k)}{\partial k}>0$, $\tau_m(q^*(k))>0$, and $\frac{\partial q_c^*(q_i^*(k))}{\partial k}-\frac{\partial q_i^*(k)}{\partial k}=\frac{\xi }{8 \gamma  k^2 (\varphi_v-\varphi_m) \sqrt{\frac{\xi }{k}}}>0$.

	Since $\frac{\partial}{\partial k}\left(\frac{\tau_m(q^*(k))+\phi}{2\phi}\right)=\frac{1}{2\phi}\frac{\partial \tau_m(q^*(k))}{\partial k}<0$, the probability that the challenger (incumbent) wins the election decreases (increases) as $k$ increases. Both policies $q_i^*(k)$ and $q_c^*(q_i^*(k))$ increase with $k\geq\bar k$, with $\lim_{k \to \infty}q_i^*(k)=\lim_{k \to \infty}q_c^*(q_i^*(k))=\varphi_v$. Since $\varphi_v>q_i^*(k) > q_c^*(q_i^*(k))$ for every finite $k\geq\bar k$, the voter always prefers policy $q_i^*(k)$ to $q_c^*(q_i^*(k))$. Moreover, the expected gains from quality are increasing in $k$ since $\frac{\partial}{\partial k}\left(\frac{\phi^2-\tau_m^2(q^*(k))}{4\phi}\right)= -\frac{\tau_m(q^*(k))}{2\phi}\frac{\partial \tau_m(q^*(k))}{\partial k}>0$. Therefore, as $k$ increases, the voter has better policy proposals, a higher probability of implementing her favorite policy, and higher expected gains from quality. It follows that, for every finite $k\geq \bar k$, $\frac{\partial W_v^*(k)}{\partial k}>0$. The proof is completed by noting from equation \eqref{eq:welfare} that $\lim_{k\to\infty}W_v^*(k)=\hat W_v=\phi/4$.
\end{proof}

\begin{proof}[\bf Proof of Corollary~\ref{cor:nomedia}]
	In the absence of a media outlet, the median voter theorem holds and both candidates offer $\varphi_v$. The voter, being uninformed, cannot do better than select candidates randomly given proposals $q=(\varphi_v,\varphi_v)$. Therefore, the voter's expected payoff without a media outlet is $-\gamma(\varphi_v-\varphi_v)^2+\mathbb{E}_f[\theta]=0$. From Proposition~\ref{prop:welfare} we have that the welfare of the voter is at its minimum for $k\in\left(0,\bar k/4\right]$. Moreover, $W_v^*(k)$ is continuous and increasing in $k$, and $W_v^*(k)=-\gamma(\varphi_v-\varphi_m)^2+\phi/4$ for all $k\in\left(0,\bar k/4\right]$. Therefore, if $-\gamma(\varphi_v-\varphi_m)^2+\phi/4<0$ there exists a $k'>\bar k/4$ such that $W_v^*(k)< 0$ for all $k\in\left(0,k'\right)$ and $W_v^*(k)>0$ for all $k>k'$.
\end{proof}

\begin{proof}[\bf Proof of Lemma~\ref{lemma:increg}]
	The proof follows directly from maximizing $\iota(k)=\frac{\phi - \tau_m(q^*(k))}{2\phi}$ with respect to $k$, where $q^*(k)$ is as in Proposition~\ref{prop:eqmpol}. Since $\tau_m(q^*(k))>0$ for every finite $k>\bar k/4$ and $\tau_m(q^*(k))=0$ for every $k\in\left(0,\bar k/4\right]$, it follows that $\iota(k)$ is maximized in $k\in\left(0,\bar k/4\right]$, where $\iota(k)=\frac{1}{2}$. Moreover, since $\lim_{k\to\infty} \tau_m(q^*(k))=0$, then $\lim_{k\to\infty} \iota(k)=1/2$.
\end{proof}

\section{Supplementary Appendix}\label{app:suppl}

\subsection{Ideological Media}

In this section, I consider the case where the media outlet's political gains are purely policy dependent. Consider the media outlet's preferences over candidates as stated in Section~\ref{sec:model}. Given the realized state $\theta$, the outlet payoff from persuading the voter is $\theta-\tau_m(q)$ if $\theta\in(\tau_v(q),\tau_m(q))$ and $\tau_m(q)-\theta$ if $\theta\in(\tau_m(q),\tau_v(q))$. From the candidates' perspective, the expected amount of resources that---given policies $q$---the media outlet can use for misreporting information is therefore $|\tau_m(q)|$, as $\mathbb{E}_f[\theta]=0$. Hereafter, I indicate the media outlet's expected political gains with $\hat\xi(q)=|\tau_m(q)|$. When $q_i=q_c$ the media outlet always reports truthfully the value of the state $\theta$ as there are no gains from persuasion (this is true also at the interim stage). When instead policies are different from each other, the outlet has an incentive to persuade the voter that is proportional to $|\tau_m(q)|$. Notice that, given any pair of policies $q$, the analysis of the communication subgame is unaltered.

I now turn to study the policy-making stage when $\hat\xi(q)=|\tau_m(q)|$. Hereafter, I refer to the analysis previously performed in Appendix~\ref{app:policystage}. Consider the challenger's ``best response to the left,'' with $q_c \leq q_i$. In this case, we have that $\hat\xi(q)=\tau_m(q)>0$. Recall that, in equilibrium, the challenger wins the election when $\theta<h(r^*(q))$, where $h(r^*(q))=\tau_v(q)+\tfrac{1}{2}\sqrt{\tau_m(q)/k}$. Differently than the case where the media outlet's gains $\xi$ are fixed, now the derivative $\frac{d h(r^*(q))}{d q_c}$ is not always positive. Since $h(r^*(q))$ is concave in $q_c$, $h(r^*(\cdot))$ has a maximum when $q_c$ is such that $\frac{d h(r^*(q))}{d q_c}=0$. Therefore, the challenger maximizes her payoff when $q_c$ is such that either $\frac{d h(r^*(q))}{d q_c}=0$ and $h(r^*(q))<\tau_m(q)$, or such that $\frac{d h(r^*(q))}{d q_c}>0$ and $h(r^*(q))=\tau_m(q)$. In particular, we have that
\[
\frac{d h(r^*(q))}{d q_c}= \underbrace{2\gamma(\varphi_v - q_c)}_{\geq0} + \underbrace{\frac{1}{4k}\sqrt{\frac{k}{\tau_m(q)}} 2\gamma(\varphi_m-q_c)}_{\leq0}.
\]
When $\frac{d h(r^*(q_i,q_c^*))}{d q_c}=0$, then the solution $q_c^*$ implicitly satisfies 
\[
\frac{\varphi_v-q_c^*}{q_c^*-\varphi_m} =\frac{1}{4k}\sqrt{\frac{k}{\tau_m(q_i,q_c^*)}}.
\]
The left-hand side of the above equation is always positive for $q_c\in (\varphi_m,\varphi_v)$, it is decreasing in $q_c$, and it has a vertical asymptote at $\varphi_m$. In particular, the LHS goes to $\infty$ as $q_c \downarrow \varphi_m$. The right-hand side of the above equation is always positive, increasing in $q_c$, and it has a vertical asymptote such that the RHS goes to $\infty$ as $q_c \uparrow q_i$. Therefore, the solution to the above equation is some $q_c^*\in (\varphi_m,q_i)$ such that $V_c(q_i,q_c^*)=h(r^*(q_i,q_c^*))>0$.

Suppose that the challenger's best response ``from the left'' is some $q_c^*\leq q_i$ such that $h(r^*(q_i,q_c^*))=\tau_m(q_i,q_c^*)$. In this case, there are two solutions to the condition $h(r^*(q))=\tau_m(q)$: the first is the trivial one where $q_c=q_i$, and $V_c(q)=h(r^*(q))=\tau_m(q)=0$; the second (implicit) solution is again some $q_c^*\in(\varphi_m,q_i)$. To see this, notice that $\tau_m(q)$ is positive and continuously decreasing in $q_c$, with $\tau_m(q_c=q_i)=0$. Instead, $h(r^*(q))$ has a global maximum at some $q_c^h<q_i$, with $h(r^*(q_i,q_c^h))>0$ while\footnote{Notice that $h(r^*(q_i,q_c=\varphi_m))=(\varphi_m-q_i)\left[\gamma(\varphi_v-\varphi_m+\varphi_v-q_i)+\sqrt{\gamma/4k}\right]<0$.}  $h(r^*(q_i,q_c=\varphi_m))<0$. Therefore, if $h(r^*(\cdot))$ and $\tau_m(\cdot)$ cross at a $q_c^*$ such that $\frac{d h(r^*(q_i,q_c^*))}{d q_c}>0$, then it must be that $q_c^*\in(\varphi_m,q_i)$ and $V_c(q_i,q_c^*)=h(r^*(q_i,q_c^*))>0$. This second solution gives a higher payoff to the challenger than $q_c=q_i$. Therefore, $q_c=q_i$ is never a best response ``to the left,'' and $BR^L_c(q_i)$ always ensures the challenger a second mover advantage as $V_c(BR^L_c(q_i),q_i)>0$. In contrast with the case where political gains $\xi$ are fixed, here the challenger does not offer the outlet's preferred policy $\varphi_m$ even when $k$ is relatively small.

The analysis of the ``best response to the right'' $BR^R_c(q_i)$, where $q_c \geq q_i$, is similar (see also Appendix~\ref{app:best}). In this case, we have that $\hat \xi(q) =-\tau_m(q)$ and, in equilibrium, the challenger wins when $\theta<l(r^*(q))=\tau_v(q)-\tfrac{1}{2}\sqrt{-\tau_m(q)/k}$. Therefore, as in the analysis in Appendix~\ref{app:best}, we still have that $\frac{d l(r^*(q))}{d q_c}=\frac{d h(r^*(q))}{d q_c}$. 

As in the case where political gains are fixed, the incumbent seeks to minimize the challenger's expected payoff. From the above analysis, we know that the challenger can secure an expected payoff of $V_c(q_i,BR_c(q_i))>0$ with a $q_c=BR_c(q_i)\neq q_i$. By proposing policies that are closer to $\varphi_v$, the incumbent reduces the challenger's payoff from proposing $BR^R_c(q_i)$ but increases the challenger's payoff from proposing $BR_c^L(q_i)$. To see this, consider $q_i'>q_i$ such that $\varphi_m\leq q_i<q_i'\leq\varphi_v$. When the incumbent proposes $q_i'$, the challenger has the option of proposing a $q_c'\in(BR^L_c(q_i),q_i')$ such that $q_i'-q_c'=q_i-BR^L_c(q_i)$. This strategy maintains unaltered the size of conflict of interest $|\tau_v(q)-\tau_m(q)|\propto|q_c-q_i|$, it increases the outlet's political gains $\hat\xi(q)=\tau_m(q)$, and it also increases the voter's threshold $\tau_v(q)$. Therefore, $V_c(q_i',q_c')=h(r^*(q_i',q_c'))>h(r^*(q_i,BR^L_c(q_i)))=V_c(q_i,BR^L_c(q_i)))>0$. A similar argument holds for the challenger's best response to the right $BR^R_c(q_i')$.

The extent to which the media outlet can persuade the voter depends on the ratio $\hat\xi(q)/k$. When misreporting costs grow arbitrarily large, $k\to+\infty$, the outlet always reveals the state and the policies of both candidates converge to $\varphi_v$. By contrast, when $k\to 0^+$, the policies of both candidates converge to $\varphi_m$. To see this, suppose that $q_i=\varphi_m$ and $q_c\in(\varphi_m,\varphi_v]$. No matter how close $q_i$ and $q_c$ are, the outlet obtains full persuasion as $k\to 0^+$, and therefore the incumbent wins when $\theta>l(r^*(\varphi_m,q_c))=\tau_m(\varphi_m,q_c)$. Since $\tau_m(\varphi_m,q_c)<0$, the expected payoff of the challenger is lower than the payoff she would get by proposing $\varphi_m$ as well, $V_c(\varphi_m,q_c)<0=V_c(\varphi_m,\varphi_m)$. A similar argument applies to the case where $q_c=\varphi_m$ and $q_i\in(\varphi_m,\varphi_v]$. Therefore, when $k\to 0^+$, candidates best reply to each other by proposing $\varphi_m$.

We obtain that equilibrium policies converge to $\varphi_v$ as $k\to+\infty$, to $\varphi_m$ as $k\to 0^+$, and diverge otherwise for intermediate levels of $k$, with $q_i^*(k)\neq q_c^*(q_i^*(k),k)$. As in the main analysis, policy divergence implies an increased conflict of interest between the media outlet and the voter, and thus more misreporting and persuasion.\footnote{Notice that, for this argument, it is irrelevant whether the challenger best responds to the left or to the right. Therefore, the results are robust to the challenger's tie breaking rule.} Therefore, the main results of this paper about the relationship between the rate of persuasion and the level of misreporting costs carry through even when using policy dependent political gains $\hat \xi(q)$ for the media outlet.  

\subsection{Equilibrium of the Communication Subgame}

\begin{Alemma}\label{lemma:perfect}
	The sender-preferred generic equilibrium of $\hat\Gamma$, where $\lambda=\tau_v$, is also the unique perfect sequential equilibrium \citep{grossman1986perfect} of $\hat\Gamma$.
\end{Alemma}
\begin{proof}
	First, since the concept of perfect sequential equilibrium (PSE) is stronger than the Intuitive Criterion, every PBE of the communication subgame that is eliminated by the Intuitive Criterion is not perfectly sequential. Consider now a generic equilibrium of $\hat \Gamma$ with $\lambda\in\left[\tau_v,\tau_v+\tfrac{1}{2}\sqrt{\xi/k}\right]$ as in Proposition~\ref{prop:monopoly}. I define $K(r,\lambda)$ as the set of types that can potentially gain by deliverning an off-path report $r\in(l(\hat r(\lambda)),\hat r(\lambda))$. Formally, $K(r,\lambda)=\left( l(r), \vartheta(r)\right)$, where $\vartheta(r)=\left\{\theta'|C(r,\theta')=C(\hat r(\lambda),\theta') \text{ and } r<\theta'<\hat r(\lambda)\right\}$. Notice that $\hat r(\tau_v)\in(l(\hat r(\lambda)),\hat r(\lambda))$ for every $\lambda\in\left(\tau_v,\tau_v+\tfrac{1}{2}\sqrt{\xi/k}\right]$, and it is thus off-path in these equilibria. Given that $\mathbb{E}_f[\theta|\theta\in(l(\hat r(\tau_v)),\hat r(\tau_v))]=\tau_v$ (see Proposition~\ref{prop:monopoly}), it follows that $\mathbb{E}_f[\theta|\theta\in(l(\hat r(\tau_v)),\vartheta(\hat r(\tau_v)))]>\tau_v$. Suppose that, in a generic equilibrium of $\hat \Gamma$ with $\lambda\in\left(\tau_v,\tau_v+\tfrac{1}{2}\sqrt{\xi/k}\right]$, upon observing the off-path report $\hat r(\tau_v)$ the voter conjectures that such a report has been delivered from types of sender in the set $K(\hat r(\tau_v),\lambda)$. The voter's updated ``consistent'' beliefs $p_K$ are the conditional distribution of the outlet's type given that $\theta\in K(\hat r(\tau_v),\lambda)$. Given beliefs $p_K$, the voter casts a ballot for the incumbent, as $\mathbb{E}_{p_K}[\theta]=\mathbb{E}_f[\theta|\theta\in(l(\hat r(\tau_v)),\vartheta(\hat r(\tau_v)))]>\tau_v$. This makes a deviation to $\hat r(\tau_v)$ profitable only for types $\theta\in K(\hat r(\tau_v),\lambda)$. Therefore, all generic equilibria of $\hat\Gamma$ with $\lambda\in\left(\tau_v,\tau_v+\tfrac{1}{2}\sqrt{\xi/k}\right]$ are \emph{not} perfectly sequential. Finally, consider the generic equilibrium with $\lambda=\tau_v$. In this case, every off-path report $r'\in(l(\hat r(\tau_v)),\hat r(\tau_v))$ induces consistent beliefs such that $\mathbb{E}_{p_K}[\theta]=\mathbb{E}_f[\theta|\theta\in(l(r'),\vartheta(r'))]<\tau_v$, and the voter casts a ballot for the challenger after observing $r'$. No type of sender would benefit from delivering any off-path report $r'\in(l(\hat r(\tau_v)),\hat r(\tau_v))$, and thus this equilibrium is perfectly sequential. Therefore, the generic equilibrium of $\hat\Gamma$ with $\lambda=\tau_v$ is the only perfect sequential equilibrium of $\hat\Gamma$.
\end{proof}

\begin{Alemma}\label{lemma:undefeated}
	The undefeated \citep{mailath1993belief} generic equilibrium of $\hat\Gamma$ is unique and it has $\lambda=\tau_v$.
\end{Alemma}
\begin{proof}
	Denote the set of generic equilibria of $\hat\Gamma$ with $GE(\hat\Gamma)$ (see Proposition~\ref{prop:monopoly}). With a slight abuse of notation, say that $u_m(\rho,p,\lambda,\theta)$ denotes the media outlet's payoff in the generic equilibrium $(\rho,p)$ with parameter $\lambda$ and when the state is $\theta$. From \cite{mailath1993belief}, $(\rho,p)\in GE(\hat \Gamma)$  defeats $(\rho',p')\in GE(\hat\Gamma)$ if there exists a report $r\in\Theta$ such that (i) $\rho'(\theta)\neq r$ for every $\theta\in\Theta$, and $\kappa(r)\equiv\{ \theta\in\Theta|\rho(\theta)=r\}\neq\varnothing$; (ii) for every $\theta\in\kappa(r)$, $u_m(\rho,p,\lambda,\theta)\geq u_m(\rho',p',\lambda,\theta)$, with the inequality being strict for some $\theta\in \kappa(r)$; (iii) there is some $\theta\in\kappa(r)$ such that $p'(\theta|r)$ is different from the conditional probability that the state is in $\kappa(r)$.
	
	Consider now the generic equilibrium of $\hat\Gamma$ with $\lambda=\tau_v$, and notice that $\hat r(\tau_v)$ is off-path in every other generic equilibrium of $\hat\Gamma$ with $\lambda'\in\left(\tau_v,\tau_v+\tfrac{1}{2}\sqrt{\xi/k}\right]$. For every $\theta\in (l(\hat r(\tau_v)),\hat r(\tau_v))$, we have that $u_m(\cdot,\lambda=\tau_v,\theta)=\xi-kC(\hat r(\tau_v),\theta)>u_m(\cdot,\lambda',\theta)$ as either $u_m(\cdot,\lambda',\theta)=0$ or $u_m(\cdot,\lambda',\theta)=\xi-kC(\hat r(\lambda'),\theta)$. Finally, $\mathbb{E}_f[\theta|\theta\in\kappa(\hat r(\tau_v))]=\tau_v$, and therefore $p'(\theta|\hat r (\tau_v))$ must differ from the conditional probability of $\theta\in\kappa(\hat r(\tau_v))$ for the voter to select the challenger upon observing $\hat r(\tau_v)$ in generic equilibria with $\lambda'\neq\tau_v$. Hence, the generic equilibrium of $\hat\Gamma$ with $\lambda=\tau_v$ defeats every other generic equilibria of $\hat\Gamma$. 
	
	Reports in $(l(\hat r(\tau_v)),l(\hat r(\lambda')))$ are the only reports that are on-path in a generic equilibrium of $\hat\Gamma$ with $\lambda'\in\left(\tau_v,\tau_v+\tfrac{1}{2}\sqrt{\xi/k}\right]$ but not when $\lambda=\tau_v$. In all these generic equilibria, upon observing a $r'\in(l(\hat r(\tau_v)),l(\hat r(\lambda')))$ the voter casts a ballot for the challenger. Therefore, no generic equilibrium of $\hat\Gamma$ with parameter $\lambda'$ can defeat the one with $\lambda=\tau_v$. Hence, the generic equilibrium of $\hat\Gamma$ with parameter $\lambda=\tau_v$ is the only undefeated generic equilibrium of the communication subgame. 
\end{proof}

\begin{Acorollary}\label{cor:payoffset}
	The set of equilibrium payoffs that the voter can obtain in a PBE robust to the Intuitive Criterion is $\mathcal{W}(k)=\left[W_v^*(k),\hat W_v\right]$, where $W_v^*(k)$ is as in equation \eqref{eq:welfare} and $\hat W_v=\phi/4$ is the full-information welfare.
\end{Acorollary}
\begin{proof}
	By definition, equation~\eqref{eq:welfare} describes the lowest payoff the voter can receive in a PBE robust to the Intuitive Criterion. As assumed in Appendix~\ref{app:policy}, suppose that the challenger selects the voter's least preferred policy when indifferent, and consider a generic equilibrium of $\hat\Gamma$ as in Proposition~\ref{prop:monopoly}. By the continuity of $l(r)$ and $h(r)$ with respect to $r$, and of $r^*(\lambda)$ with respect to $\lambda$, we obtain that the voter's equilibrium welfare is continuously (weakly) increasing in $\lambda\in\left[\tau_v(q),\tau_v(q)+\tfrac{1}{2}\sqrt{\xi/k}\right]$: for higher $\lambda$, the set of states in which persuasion occurs (weakly) shrinks and both the incumbent and the challenger's policies get (weakly) closer to the voter's bliss policy $\varphi_v$. When $\lambda=\tau_v(q)+\tfrac{1}{2}\sqrt{\xi/k}$ there is no persuasion at all, and the voter always votes for her preferred candidate \emph{as if} under complete information. Since the media outlet has no persuasive power, the median voter theorem holds and both candidates propose $\varphi_v$. In this case the voter's welfare is $\hat W_v=\phi/4$, and therefore in a PBE robust to the Intuitive Criterion the voter can obtain any payoff in the set $\left[W_v^*(k),\hat W_v\right]$.
\end{proof}

\subsection{Simultaneous Policy-making}

\begin{Alemma}\label{lemma:simultaneous}
	Suppose that candidates propose policies simultaneously. Then, if $k\leq \bar k/4$ there is an equilibrium where both policies converge to $\varphi_m$.
\end{Alemma}
\begin{proof}
	Consider a variation of the model in Section~\ref{sec:model} such that candidates propose policies simultaneously rather than sequentially. Suppose that there is an equilibrium where $q'=(\varphi_m,\varphi_m)$, and notice that $\tau_v(q')=\tau_m(q')=V_j(q')=0$, $j\in\{i,c\}$. A deviation by the challenger to some $q_c>\varphi_m$ cannot be profitable if $V_c(\varphi_m,q_c)=l(r^*(\varphi_m,q_c))\leq 0$, where $l(r^*(\varphi_m,q_c))=\max \{  r^*(\varphi_m,q_c) - \sqrt{\xi/k}, \tau_m(\varphi_m,q_c)\}$ and $r^*(\varphi_m,q_c)=\min\{ \tau_v(\varphi_m,q_c) + \tfrac{1}{2}\sqrt{\xi/k},2\tau_v(\varphi_m,q_c) - \tau_m(\varphi_m,q_c)\}$. Since for a $q_c\in(\varphi_m,\varphi_v)$ we have that $\partial \tau_m(\varphi_m,q_c) / \partial q_c < 0 < \partial \tau_v(\varphi_m,q_c)/\partial q_c$, and $\partial \tau_v(\varphi_m,q_c)/\partial q_c=0$ when $q_c=\varphi_v$, the challenger's payoff $V_c(\varphi_m,q_c)$ has a maximum either at $q_c=\varphi_m$ or at $q_c=\varphi_v$. The condition $V_c(\varphi_m,\varphi_v)\leq 0$ gives us $k\leq \xi / 4\gamma^2 (\varphi_v - \varphi_m)^4 = \bar k/4$. Therefore, when $k\leq \bar k/4$ there is an equilibrium where policies converge to $(\varphi_m,\varphi_m)$. 
\end{proof}

\begin{Alemma}\label{lemma:nopure}
	Suppose that candidates propose policies simultaneously. Then, there is no pure strategy equilibrium in the policy-making stage when $k\geq \bar k$.
\end{Alemma}
\begin{proof}
	Suppose that candidates propose policies simultaneously and that $k\geq \bar k$. Since candidates are symmetric, their best response functions are symmetric as well. From Proposition~\ref{prop:bestresp} the best response of the challenger to a $q_i$ is either $\varphi_v$ or $q_i-\eta(k)$.  Since $q_j-\eta(k)$ is always below the $q_j$ line for every finite $k\geq \bar k$, a pure strategy equilibrium of the policy-making stage must have one of the two candidates proposing $\varphi_v$. Suppose that $q_i=\varphi_v$, and therefore $BR_c(\varphi_v)=\varphi_v-\eta(k)$. From Proposition~\ref{prop:bestresp} we obtain that $BR_i(BR_c(\varphi_v))=\varphi_v$ only if $\varphi_v-\eta(k)\leq \varphi_v + \eta(k)-\sqrt[4]{\xi/\gamma^2k}$. By rearranging and plugging $\eta(k)=\sqrt{\xi/k}/4\gamma(\varphi_v-\varphi_m)$, we obtain the condition $\sqrt[4]{\xi/\gamma^2k}\leq \sqrt{\xi/k}/2\gamma(\varphi_v-\varphi_m)$. Such a condition is not satisfied at $k=\bar k$, and therefore it is not satisfied for any finite $k>\bar k$.
\end{proof}


\addcontentsline{toc}{section}{References}
\bibliographystyle{apacite}
\bibliography{biblio_news}
\end{document}